\documentclass[12pt,twoside,a4paper]{amsart}
\usepackage[english]{babel}
\usepackage[utf8]{inputenc}

\usepackage{ulem}

\usepackage{amsmath,amssymb,amsthm,amscd,appendix,bbold,dsfont,epsfig,eucal,enumitem,color,graphics,graphicx,indentfirst,latexsym,stmaryrd,tikz,ulem,verbatim,wrapfig}
\usepackage[colorlinks=true]{hyperref}

\usepackage{nameref}

\makeatletter
\let\orgdescriptionlabel\descriptionlabel
\renewcommand*{\descriptionlabel}[1]{%
  \let\orglabel\label
 \let\label\@gobble
  \phantomsection
  \edef\@currentlabel{#1}%
 \let\label\orglabel
  \orgdescriptionlabel{#1}%
 }
 \makeatother

\usepackage{tikz}
\newcommand{\potentielL}[2]
{\draw [domain = 0 : {#2},samples =200 ] plot({\x+#1} ,{#2*#2- #2*#2*cos(2*\x*pi/#2 r)});
}

\DeclareMathOperator*{\argmin}{arg\,min}

\begin{document}

\numberwithin{equation}{section}

\theoremstyle{plain}
\newtheorem{theorem}{Theorem}
\newtheorem{lemma}[theorem]{Lemma}
\newtheorem{proposition}[theorem]{Proposition}
\newtheorem{corollary}[theorem]{Corollary}
\newtheorem{example}[theorem]{Example}

\theoremstyle{definition}
\newtheorem{definition}[theorem]{Definition}
\newtheorem{notation}[theorem]{Notation}
\newtheorem{assumption}[theorem]{Assumptions}
\newtheorem{hypotheses}[theorem]{Hypotheses}
\newtheorem{remark}[theorem]{Remark}

\newcommand{\osc}{\text{\rm var}}
\newcommand{\supp}{\text{\rm supp}}
\newcommand{\card}{\text{\rm card}}
\newcommand{\dist}{\text{\rm dist}}
\newcommand{\GS}{\Omega_{\textrm{\rm GS}}}

\title[Classification of discrete weak KAM solutions]{Classification of discrete weak KAM solutions on linearly repetitive quasi-periodic sets}

\author{Eduardo Garibaldi} 
\address{IMECC, Universidade Estadual de Campinas, 13083-859, Cam\-pi\-nas, Bra\-sil}
\email{garibaldi@ime.unicamp.br}

\author{Samuel Petite} 
\address{LAMFA,  CNRS, UMR 7352, Universit\'e de Picardie Jules Verne, 80 000 Amiens, France}
\email{samuel.petite@u-picardie.fr}

\author{Philippe Thieullen} 
\address{Institut de Math\'ematiques de Bordeaux, CNRS, UMR 5251, Universit\'e de Bordeaux, F-33405 Ta\-len\-ce, France}
\email{Philippe.Thieullen@u-bordeaux.fr}

\thanks{The authors acknowledge the financial support of FAPESP grant 2019/10485-8 and  ANR project IZES ANR-22-CE40-0011.}

\subjclass[2020]{Primary: 82D25, 82M30; Secondary: 37J51,  52C23}
\keywords{discrete weak KAM solution, linear repetitivity, calibrated configuration, ground state, minimizing configuration, Frenkel-Kontorova model} 
\maketitle

\begin{abstract}
A discrete weak KAM solution is a potential function that highlights the ground state configurations at zero temperature of an infinite chain of atoms interacting with a periodic or quasi-periodic substrate. It is well known that weak KAM solutions exist for periodic substrates as in the Frenkel-Kontorova model. Weak solutions may not exist in the almost periodic setting as in the theory of stationary ergodic Hamilton-Jacobi equations (where they are called correctors).  
For linearly repetitive quasi-periodic substrates, we show that equivariant interactions that fulfill a twist condition and a non-degenerate property always admit sublinear weak KAM solutions.
We moreover classify all possible types of weak KAM solutions and calibrated configurations according to an intrinsic prefered order. The notion of prefered order is new even in the classical periodic case.

\end{abstract}

\section{Introduction}
We consider a generalized model  of Frenkel-Kontorova type on the real line. The model describes the states at equilibrium of   chains of atoms  interacting with their nearest neighbors  and with an underlying one-dimensional substrate. The interaction between two successive atoms at the positions $(x_n,x_{n+1})$ has the general form $E(x_n,x_{n+1})$ for some continuous function $E : \mathbb{R} \times \mathbb{R} \to \mathbb{R}$,  called subsequently {\it interaction model}. 

The seminal Frenkel-Kontorova model \cite{FrenkelKontorova1938} was proposed to represent the dislocation in a crystal.
This simple model has appeared to be universal to describe several physical problems \cite{AubryLeDaeron1983,BraunKivshar2004}. 
Its standard version may be given as
\begin{gather}
E(x,y) = \frac{1}{2\tau} (y-x)^2 -\lambda(y-x) +\tau K V(x), \label{Equation:FrenkelKontorovaModel}
\end{gather}
where the constants $ \lambda, K \in \mathbb{R}$, $\tau>0$, are the parameters of the model. 
Here $E$ has the physical dimension of an action assuming the mass is equal to 1. 
The constant $\tau$ plays the role of a discretized time, 
$ \tau \lambda $ is a inter-distance between two successive atoms 
(a positive constant $\lambda>0$ forces the chain $(x_n)_{n\in\mathbb{Z}}$ to be increasing),   
$ K$ is a dimensionless constant that measures the strength of the interaction between one atom of the chain and the substrate,  
and $V$ is a periodic potential of unit size describing a periodic external environment.
Our main motivation is to understand interaction models of Frenkel-Kontorova type without assuming $ V $ periodic.

A central part of the theory consists in studying the set of configurations at the ground state, called hereinafter {\it Ma\~n\'e calibrated} configurations, that is  
the set of positions of the atoms in a chain $(x_n)_{n\in\mathbb{Z}}$  that minimize, in a sense to be defined,  the total action 
\[
\argmin_{(x_n)_{n \in \mathbb{Z}}} \ \sum_{n\in\mathbb{Z}} E(x_n,x_{n+1}).
\]
A key tool for this study is the notion of weak KAM solution at an appropriate action level. 
Let us recall these concepts.
The {\it ground action} or atomic mean action is a particular choice of action level defined as
\begin{gather}
\bar E := \lim_{n\to+\infty} \inf_{x_0, \ldots,x_n \in \mathbb{R}} \ \frac{1}{n} \sum_{k=0}^{n-1} E(x_k,x_{k+1}). \label{Equation:GroundActionLevel}
\end{gather}
It is finite by the choice of the interaction model we are going to make.
A continuous function $u:\mathbb{R} \to \mathbb{R}$ is called  \textit{weak KAM solution} if it satisfies
\begin{gather}
\begin{cases}
\forall \,x,y \in \mathbb{R}, \quad u(y)-u(x) \leq E(x,y)-\bar E, \\
\forall \, y, \,\,\, \exists \, x \quad \text{ s.t. } \quad u(y)-u(x) = E(x,y)-\bar E.
\end{cases} \label{Equation:DefinitionWeakKAMsolution}
\end{gather}

The relevance of these solutions has been highlighted by Fathi in~\cite{Fathi1997a,Fathi1997b,Fathi2010} in the context of the Hamilton-Jacobi equations on compact manifolds, 
and by Contreras, Iturriaga, Paternain, Paternain \cite{ContrerasIturriagaPaternain1998}, Con\-tre\-ras~\cite{Contreras2001}, Fathi, Maderna~\cite{FathiMaderna2007} for non compact manifolds. In~\cite{ArnaudZavidovique2023}, Arnaud and Zavidovique describe precisely the pseudograph set of $\frac{du}{dx}$ for exact symplectic twists maps of the annulus, where $E(x,y)$ is understood as a global generating function.
These works are influenced by the viewpoint proposed by Mather \cite{Mather1991, Mather1993} for studying minimizing orbits of Lagrangian systems. 

A discretization of weak KAM theory applied to optimal transportation also enables to relate deep results of existence of optimal transport maps \cite{FathiFigalli,BernardBuffoni}.
More recently, the discrete analogue of the Hamilton-Jacobi equations has been studied by many authors~-- 
see, for instance, \cite{Gomes2005, GaribaldiThieullen2011, Zavidovique2012_01}. 
In particular, it is known there exists a weak KAM solution $u$ for periodic interaction models, that is, for models satisfying 
\[\forall x,y \in \mathbb{R}, \quad E(x+1, y+1) = E(x,y), 
\]
as is the case of the standard Frenkel-Kontorova model 
(for details see~\cite{ChouGriffiths1986, GaribaldiThieullen2011}). Moreover, such a solution $u$ may be chosen periodic and hence bounded. 

Regarding the atomic interpretation,
an interest of weak KAM solution lies in the concept of  
 $u$-{\it calibrated} configuration $(x_n)_{n\in\mathbb{Z}}$, that is: 
 \begin{align}\label{eq:uCalibratedConf}
\forall \,m,n \in \mathbb{Z}, \ m<n, \qquad  \sum_{k=m}^{n-1} \left( E(x_k,x_{k+1}) - \bar E \right) = u(x_{n})-u(x_m).
\end{align}
A weak KAM solution $u$ plays the role of a fictitious potential needed to keep a finite block of atoms of a configuration in equilibrium by applying opposing forces at both ends, simulating the remaining action of the chain outside the block \cite{ChouGriffiths1986}.
However the effect of $u$ as a fictitious potential does not occur everywhere in space. A $u$-calibrated configuration enables to materialize these positions. 

A weaker notion of calibration may be defined which we call Ma\~n\'e calibration.
More precisely, the configuration $(x_n)_n$ is said to be \textit{Ma\~n\'e calibrated} if it satisfies $\forall \,m,n \in\mathbb{Z}$,  $n\geq1$,
\begin{align}\label{eq:ManeCalibratedConf}
 \sum_{k=m}^{m+n-1} \left(E(x_k,x_{k+1}) - \bar E \right) = S(x_m, x_{m+n}),
\end{align}
 where $S$ denotes the  \textit{Ma\~n\'e potential},
 \[
S(x,y) = \inf_{n\geq1} \inf_{x= x_0,\ldots, x_n=y} \sum_{k=0}^{n-1} \big( E(x_k,x_{k+1}) - \bar E \big).
 \]
In  words, the Ma\~n\'e potential between two sites measures the minimal reduced action necessary to go from one site to another. A Ma\~n\'e calibrated configuration is an infinite chain such that each finite sub-chain realizes the smallest reduced action between its two endpoints. 

We focus on the ground action $\bar E$, defined in~\eqref{Equation:GroundActionLevel}. Other calibration levels may be chosen as in~\cite{Zavidovique2012_01}, where weak KAM solutions are constructed for some non-explicit constants. In order to  stay in the class of weak KAM solutions that are sublinear at infinity, $\bar E$  is the only possible level (for details see appendix~\ref{nivel unico}). The interaction model $E$ is supposed to satisfy usual properties: $C^2$, twist, and superlinear as described more precisely in hypotheses~\ref{assumption:interactionModel}. 
Actually, our main results link sublinear growths of weak KAM solutions to a non-degeneracy condition:  $ \inf_{x \in \mathbb R} E(x,x) > \bar E $.  The complementary  case $\inf_{x \in \mathbb R} E(x,x) = \bar E$ is treated in section~\ref{periodic and linear}. Previous works assume the substrate to be periodic. The purpose of this article is to extend the theory of calibrated configurations and weak KAM solutions to quasi-periodic substrates. Here the interaction model $ E $ is supposed to be {\it equivariant with respect to a linearly repetitive quasi-periodic set} (see definitions~\ref{def:faussedefLR} and~\ref{invariancia de motivos}), a property which naturally appears when considering the classical examples of quasicrystals.
We classify thus all weak KAM solutions in three types depending of their growths at infinity: type I, the sublinear case, type II, the linear case,  and type III, a mix of these two growths.

\begin{theorem}\label{theo:mainThm0}
Any superlinear and weakly twist interaction $E$ fulfilling  $ \inf_{x} E (x, x)~>~\bar E $ always admits a Lipschitz weak KAM solution. 

If $ E $ is in addition pattern equivariant with respect to a linearly repetitive quasi-periodic set, then there is a selected half-space (either the positive or the negative reals), depending only on $E$, 
such that any weak KAM solution $ u $ belongs to one of  three non-empty disjoint classes,
namely, 
\begin{itemize}
\item type I:
$ u $ has a sublinear growth on the real line,
\item type II:
$ u $ has a linear growth on the real line, 
\item  type III:
$ u $ has a mixed linear/sublinear growth: $u$  grows sublinearly  on the  selected half-space  and  grows  linearly on the complementary  half-space.   
\end{itemize}
 Weak KAM solutions of type III are the only ones for which there is no bi-infinite calibrated configurations as in~\eqref{eq:uCalibratedConf}.
 
Besides, any two weak KAM solutions of the same type lie at uniform distance from each other.
\end{theorem}

Our framework includes significant  and classical quasicrystals such as Fibonacci, substitutive ones or  typical cut-and-project quasicrystals~\cite{Alisteetal}. 
One-dimensional quasicrystals are essentially discrete sets for which any finite pattern repeats  in  space in a syndetic manner and with some pattern-dependent density. 
For models on quasicrystals, the atoms  of the chain are supposed to interact with an underlying aperiodic substrate, a structure having no translational symmetry but exhibiting a long-range order~\cite{Shechtmannetal}. 
This context is analogous but more rigid than the almost-periodic scenery (e.g. a Frenkel-Kontorova model with a sum of incommensurable periodic functions as  potential  $V$). A major difference between the two models is that, in our setting, the interaction reaches its minimum. This avoids pathological cases.

Actually this quasi-periodic context fall into the framework of  a topological stationary setting (see \cite{LionsSouganidis2003,GambaudoGuiraudPetite2006,GaribaldiPetiteThieullen2017}). 
Initial studies  (see, for instance, \cite{GambaudoGuiraudPetite2006,Trevino, JianxingXifeng}) showed that some classical properties of the periodic case still hold when models on quasicrystals are taken into account, in particular the links between the minimizing configurations and their rotation numbers. Notably, unbounded  Ma\~n\'e calibrated configurations do exist as shown in \cite{GaribaldiPetiteThieullen2017}. However all these previous works did not answer to the existence problem of a weak  KAM solution.

Our proof techniques remain mainly in the territory of classical analysis.  Partial results where obtained in  \cite{GaribaldiPetiteThieullen2017} using ergodic tools and Mather measures. The present work is actually independent  from \cite{GaribaldiPetiteThieullen2017}; it gives a simpler proof of the existence of calibrated configuration, and above all, it constructs and classifies all weak KAM solutions. 

In the Hamilton-Jacobi framework, a way to obtain weak KAM solutions on a non-compact space  pass through the use of Busemann functions associated with an appropriate potential as in  \cite{Contreras2001}. However,  this method adapted to the present context  does not  guarantee the existence of a discrete  weak KAM solution, since the second condition  in \eqref{Equation:DefinitionWeakKAMsolution} is not \textit{a priori} satisfied. We overcome this problem by studying a Lax-Oleinik operator on a suitable space of functions associated to a bi-infinite Ma\~n\'e-calibrated configuration. The properties of such configurations leads to the classification. A precise description of these relations and the classification is given in theorem~\ref{theo:mainThm},  providing qualitative criteria to distinguish to which case a weak KAM solution belongs.

In the next section, we detail our assumptions and more rigorously formulate our central results. Note that theorem~\ref{theo:mainThm0} also provides a classification on the asymptotic behavior of weak KAM solutions  in the periodic case, as summarized in theorem \ref{theo:PeriodicCase}. At our knowledge, this is a new result in this classical setting.

\section{Main results}

We consider a general interaction model $E(x,y)$ (in which each variable describes a position on the real line).

\begin{hypotheses} \label{assumption:interactionModel}
An interaction model is a $C^0$ function $E:\mathbb{R} \times \mathbb{R} \to \mathbb{R}$ that fulfills the following three conditions
\begin{enumerate}
\item[(H1)] $E(x,y)$ is \textit{locally uniformly bounded} in the sense that 
\[
\inf_{x,y \in \mathbb{R}} E(x,y) >-\infty, \qquad \sup_{x \in \mathbb R} E(x,x) <+\infty,
\]
\item[(H2)] $E(x,y)$ is \textit{locally uniformly Lipschitz} 
for every $R>0$, there exists a constant $C_\text{Lip}^E(R) >0$ such that, for every $x,y,z \in\mathbb{R}$,
\begin{align*}
&\text{if  $|y-x| \leq R, \ |z-x| \leq R$   then   $|E(x,z)-E(x,y)| \leq C_\text{Lip}^E(R) |z-y|$,} \\
&\text{if  $|z-x| \leq R, \ |z-y| \leq R$   then  $|E(y,z) - E(x,z)| \leq C_\text{Lip}^E(R) |y-x|$,}
\end{align*}
\item[(H3)] $E(x,y)$ is \textit{uniformly superlinear}:
\[
\lim_{R \to +\infty} \ \inf_{|y-x| \geq R} \frac{E(x,y)}{|y-x|} = +\infty.
\]
\end{enumerate}
An interaction model $ E(x, y) $ is said to be \textit{weakly twist} if it is a $ C^2 $ function such that
\begin{enumerate} 
\item[(H4)]
$ \forall \, x \in \mathbb R $, $ {\displaystyle \frac{\partial^2 E}{\partial x \partial y} (x, \cdot) < 0} $  a.e., and
$ \forall \, y \in \mathbb R $, $ {\displaystyle \frac{\partial^2 E}{\partial x \partial y} (\cdot, y) < 0} $ a.e.
\end{enumerate}
\end{hypotheses}

Note that the hypothesis (H4) is slightly more general than the usual twist condition stated for any $x,y$ and not only for almost every $x,y$. This gives us the advantage to treat, for instance, a model of the form $E(x,y) = W(y-x) - \lambda(y-x) +V(x)$ with $W(s)= s^4/4$. 

Here we are interested in models that take into account the interaction between 
the atoms of the chain and the underlying environment that will be modeled by a quasi-periodic substrate $ \omega $.
In particular, we focus on \textit{pattern equivariant} interactions, a notion that not only captures the dependence with respect to the quasi-periodic environment, but also introduces the concept of short-range interaction in this context.

By a one-dimensional quasi-periodic set  
we mean a discrete set $\omega \subset \mathbb{R}$
which has  \textit{finite local complexity} and  is  \textit{repetitive}.  
To introduce these notions, we will need the one of \textit{pattern}, namely, a set  $\texttt{P}$ of the form $\omega \cap I$ for some bounded open interval $I$. Two patterns $\texttt{P}$ and $\hat{\texttt{P}}$  are  \textit{equivalent} whenever one is the translated of the other one, that is: $\texttt{P} + t = \hat{\texttt{P}}$ for some $t\in \mathbb{R}$. 
A discrete set $\omega \subset \mathbb{R}$ is said to be \textit{quasi-periodic} when the following properties are satisfied:
 \begin{description}
    \item[finite local complexity\label{item:FLC}] the  set $\omega$ possesses only finitely many  equivalence classes of patterns of cardinality $2$;
    
    \item[repetitivity\label{item:Repet}] for any $R>0$, there is a constant $M(R)$ such that for any open interval $I$ of length  at least $M(R)$, $\omega \cap I$ contains a representative from each class of patterns of diameter less than~$R$.
 \end{description}

In particular, we observe that the finite local complexity  condition implies the quasi-periodic set is (uniformly) discrete. 
Moreover, $ \omega $, as a set, is unbounded from above and below by the  repetitivity condition. By this condition, each type of  pattern occurs infinitely  many times along the real line with uniformly bounded gaps between the occurences. 
A very representative class of quasi-periodic sets is formed by quasicrystals. Such  sets  are quasi-periodic ones with an additional density property on the occurences of the patterns (see appendix \ref{Nao degenerescencia}). 
The repetitivity condition  can be interpreted as a weak homogeneity property  in a   topological sense. For a  dynamical explanation, see~\cite{GambaudoGuiraudPetite2006}. Of course periodic lattices are quasi-periodic sets, but there also exist aperiodic examples ({\it i.e.} that are invariant under no translation). The simplest ones are constructed by iteration of a procedure so called {substitution}, or in a geometrical way by a cut and project scheme (see \cite{Alisteetal,Queffelec1987}). The Fibonacci quasicrystal is a famous one that can be obtained by both methods. 
Note that in the substitutive case, the parameter $M(R)$ of the repetitivity condition can be taken with a growth at most linear in $R$. Such quasi-periodic sets are called \textit{linearly repetitive}, in the following sense.
\begin{definition}\label{def:faussedefLR}
We say that a quasi-periodic  set $ \omega $ is \textit{linearly repetitive}  if the repetitivity parameter $ M(R) $ has at most  linear growth  as a function of the upper bound $ R $ for pattern diameters.
\end{definition}

Most of quasicrystals obtained by cut and project are linearly repetitive. 
In the geometrical, combinatorial and dynamical  senses,  they are the simplest examples of aperiodic  quasicrystals. 
We refer to~\cite{Alisteetal} for a survey of their properties. 

A quasi-periodic set $ \omega $ on the real line models the underlying substrate to be considered.
We now describe the kind of interaction energy $ E $ between the chain and the substrate we are interested in.
We say that an interaction potential $ V $ is pattern equivariant with respect to $ \omega $ if two potentials are the same,
$ V(x) = V(y) $, at two distinct positions $ x \neq y $ provided the relative structures $ \omega - x $ and $ \omega - y $ coincide locally.
We generalize that idea to interaction energies in the following definition.

\begin{definition}\label{invariancia de motivos}
We say that an interaction  $ E(x, y) $ is \textit{pattern equivariant} with respect to the quasi-periodic set $ \omega $ if
there exists $ \varsigma_0 > 0 $ such that for patterns $ \texttt{P} $ (of diameter greater than $ 2\varsigma_0 $), whenever $  \texttt{P} +  t  $  
is again a pattern of $ \omega $, 
$$ E(x, y) = E(x + t, y + t) \qquad \forall \, x, y \in \big[\min  \texttt{P} + \varsigma_0, \, \max  \texttt{P} - \varsigma_0 \big]. $$
\end{definition}
We refer the reader to~\cite{GambaudoGuiraudPetite2006, GaribaldiPetiteThieullen2017}  and to appendix~\ref{Nao degenerescencia} for examples of pattern equivariant interactions.
We precise our core results (theorem~\ref{theo:mainThm0}) in the following statement.
For a weakly twist interaction model that is pattern equivariant with respect to a linearly repetitive quasi-periodic set, 
we not only show the existence but we completely classify all possible types of weak KAM solutions. 

\begin{theorem}\label{theo:mainThm}
Let $ E $ be an interaction fulfilling the assumptions (H1-4) of hypotheses~\ref{assumption:interactionModel}.
Suppose that $ \inf_{x} E (x, x) > \bar E $. Then there exist positive constants $K$ and $ r < R$ such that the following holds.
\begin{enumerate}
\item There exist Ma\~n\'e calibrated configurations. All Ma\~n\'e calibrated configurations are strictly monotone and satisfy
\begin{gather*}
\forall\, k \in \mathbb{Z}, \ r \leq | x_{k+1} - x_k | \leq R.
\end{gather*}
\item There exist weak KAM solutions. Every weak  KAM solution $u$ is Lipschitz with $\text{\rm Lip}(u) \leq K$ and satisfies
\begin{gather*}
\forall \,  y \in \mathbb{R}, \quad \argmin \{  u(\cdot) + E(\cdot,y)\} \subset  [y-R, y+R]. 
\end{gather*} 
\end{enumerate}
Assume moreover that $ E $ is pattern equivariant with respect 
to a linearly repetitive quasi-periodic set. Then there exists $\gamma>0$ such that the following holds.
\begin{enumerate}
\addtocounter{enumi}{2}
\item There exists a preferred ordering of $\mathbb{R}$  
($\epsilon=1$ for the standard ordering, $\epsilon=-1$ for the reversed one) such that
every weak KAM solution $u$ belongs to one of the following three types:
\begin{itemize}
\item type I: every $u$-calibrated configuration $(x_n)_{n\in\mathbb{Z}}$ is such that  $(\epsilon x_n)_{n\in\mathbb{Z}}$ is increasing, and
$$ \lim_{x\to \pm\infty} \frac{u(x)}{x} = 0; $$
\item type II:  every $u$-calibrated configuration $(x_n)_{n\in\mathbb{Z}}$ is such that  $(\epsilon x_n)_{n\in\mathbb{Z}}$ is decreasing, and
$$ 
\limsup_{x\to + \infty} \frac{u(\epsilon x)}{|x|} \leq -\gamma, \quad \liminf_{x\to-\infty} \frac{u(\epsilon x)}{|x|} \geq \gamma; $$
\item type III:  there is no bi-infinite $u$-calibrated configuration  and 
$$ \limsup_{x\to + \infty} \frac{u(\epsilon x)}{|x|} \leq -\gamma, \quad  \lim_{x\to-\infty} \frac{u(\epsilon x)}{|x|} = 0. $$
\end{itemize}
\item There exist  weak KAM solutions $u$ of the three types previously described. In type I and II, there exist $u$-calibrated configurations. In type III, there is no $u$-calibrated configurations.
\item Any two weak KAM solutions $ u $ and $ v $ of the same type lie at uniform distance from each other: $\sup_x |u(x) - v(x) | < +\infty $.
\end{enumerate}
\end{theorem}
To illustrate theorem~\ref{theo:mainThm}, the figures~\ref{Fig:possible_1}, \ref{Fig:possible_2}, \ref{Fig:possible_3}, present different possible asymptotic behaviors of weak KAM solutions. 

\begin{figure}[h]
\begin{tikzpicture}[scale=0.7]
{\draw [domain = -5 : 0,samples =100,dashed  ] plot({\x} ,{0.5*ln(-\x+1) } );
}
{\draw [domain = -5 : 0,samples =100,dashed  ] plot({\x} ,{-0.5*ln(-\x+1) } );
}
{\draw [domain = -5 : 0, samples = 200, red  ] plot( {\x} ,{0.5*sin(\x*\x*\x* pi * 3)* ln(-\x+1) } );
}

{\draw [domain = 0 : 5,samples =100,dashed  ] plot({\x} ,{0.5*ln(\x+1) } );
}
{\draw [domain = 0 : 5,samples =100,dashed  ] plot({\x} ,{-0.5*ln(\x+1) } );
}
{\draw [domain = 0 : 5, samples = 200, red  ] plot( {\x} ,{0.5*sin(\x*\x*\x* pi * 3)* ln(\x+1) } );
}
\draw[line width=1pt,->](-5,0)--(5,0);
\draw[ line width=0.5pt , ->](0,-0.5)--(0,1); 
\end{tikzpicture}
\caption{Type I: A sublinear growth at $\pm \infty$.}  \label{Fig:possible_1}
\end{figure}

\begin{figure}[h]
\begin{tikzpicture}[scale=0.4, xshift = -2cm]

\begin{scope}[xshift=-7cm]
{\draw [domain = -4.3 : 0,samples =200, red ] plot({-\x} ,{2.7/5*\x +0.75*cos(\x*\x*\x* pi*3 )-0.75});
}
{\draw [domain = 0 : 4.3,samples =200, red ] plot({-\x} ,{5/5*\x +0.75*cos(\x*\x*\x* pi*3 )-0.75});
}

\draw[line width=1pt,->](-4,0)--(4,0);
\draw[ line width=0.5pt , ->](0,-0.5)--(0,2.5); 
\draw[ line width=0.5pt ,dashed](0,0)--(-4,2);
\draw[ line width=0.5pt ,dashed](0,0)--(4,-2);
\end{scope} 
 
\node (explication2.1) at (-5,4) {$\epsilon =1$};
\node (explication2.2) at (6,4) {$\epsilon =-1$}; 

\begin{scope}[xshift=8cm]
{\draw [domain = -4.3 : 0,samples =200, red ] plot({-\x} ,{-2.7/5*\x- 0.75*cos(\x*\x*\x* pi*3 )+0.7});
}
{\draw [domain = 0 : 4.3,samples =200, red ] plot({-\x} ,{-5/5*\x -0.75*cos(\x*\x*\x* pi*3 )+0.7});
}
\draw[line width=1pt,->](-5,0)--(5,0);
\draw[ line width=0.5pt , ->](0,-0.5)--(0,2.5); 
\draw[ line width=0.5pt ,dashed](0,0)--(-5,-2.5);
\draw[ line width=0.5pt ,dashed](0,0)--(5,2.5);
\end{scope}


\end{tikzpicture}
\caption{Type II: On the left hand side, a growth at $\pm\infty$  at least negative. On the right hand side, a growth at $\pm\infty$  at least positive.}  \label{Fig:possible_2}
\end{figure}

\begin{figure}[h]
\begin{tikzpicture}[scale=0.6,xshift= -2cm]]
\begin{scope}[xshift= -2cm]
{\draw [domain = -4 : 0,samples =200, red ] plot({-\x} ,{0.4*\x +0.75*cos(\x*\x*\x* pi*3 )-0.7});
}
{\draw [domain = 0 : 4,samples =100,dashed  ] plot({-\x} ,{0.5*ln(\x+1) } );
}
{\draw [domain = 0 : 4,samples =100,dashed  ] plot({-\x} ,{-0.5*ln(\x+1) } );
}
{\draw [domain = 0 : 4, samples = 100, red  ] plot( {-\x} ,{0.5*sin(\x*\x*\x* pi * 3)* ln(\x+1) } );
}
\draw[line width=1pt,->](-4,0)--(4,0);
\draw[ line width=0.5pt , ->](0,-0.5)--(0,1.5); 
\draw[ line width=0.5pt ,dashed](0,0)--(4,-1.52);
\end{scope}

\node (explication2.1) at (0,2) {$\epsilon =1$};
\node (explication2.2) at (6,2) {$\epsilon =-1$}; 

\begin{scope}[xshift= 8cm]
{\draw [domain = -4 : 0,samples =200, red ] plot({\x} ,{0.4*\x +0.75*cos(\x*\x*\x* pi*3 )-0.7});
}
{\draw [domain = 0 : 4,samples =100,dashed  ] plot({\x} ,{0.5*ln(\x+1) } );
}
{\draw [domain = 0 : 4,samples =100,dashed  ] plot({\x} ,{-0.5*ln(\x+1) } );
}
{\draw [domain = 0 : 4, samples = 200, red  ] plot( {\x} ,{0.5*sin(\x*\x*\x* pi * 3)* ln(\x+1) } );
}
\draw[line width=1pt,->](-4,0)--(4,0);
\draw[ line width=0.5pt , ->](0,-0.5)--(0,1.5); 
\draw[ line width=0.5pt ,dashed](0,0)--(-4,-1.52);
\end{scope}
\end{tikzpicture}

\caption{On the left hand side,  a sublinear growth at $- \infty$ and a growth at $+\infty$ at least negative. On the right hand side, a sublinear growth at $+ \infty$ and a growth at $-\infty$ at least positive.} \label{Fig:possible_3}

\end{figure}

At the opposite, the figures~\ref{Fig:impossible_1}, \ref{Fig:impossible_2}, \ref{Fig:impossible_3},  present impossible asymptotic behaviors for weak KAM solutions.

\begin{figure}[h]
\begin{tikzpicture}[scale=0.5]
{\draw [domain = -5 : 0,samples =200, red ] plot({\x} ,{-0.4*\x -0.75*cos(\x*\x*\x* pi*3 )+0.7});
}
{\draw [domain = 0 : 5,samples =100,dashed  ] plot({\x} ,{0.5*ln(\x+1) } );
}
{\draw [domain = 0 : 5,samples =100,dashed  ] plot({\x} ,{-0.5*ln(\x+1) } );
}
{\draw [domain = 0 : 5, samples = 200, red  ] plot( {\x} ,{0.5*sin(\x*\x*\x* pi * 3)* ln(\x+1) } );
}
\draw[line width=1pt,->](-5,0)--(5,0);
\draw[ line width=0.5pt , ->](0,-0.5)--(0,2.5); 
\draw[ line width=0.5pt ,dashed](0,0)--(-5,1.9);
\end{tikzpicture}

\caption{At least (negative) linear growth at $-\infty$ and sublinear growth a $+\infty$.  The symmetric case with respect to the vertical axe is also impossible.} \label{Fig:impossible_1}
\end{figure}

\begin{figure}[h]
\begin{tikzpicture}[scale=0.4]
{\draw [domain = 0 : 5,samples =100,dashed  ] plot({\x} , {0.35*\x} );
}
{\draw [domain = 0 : 5,samples =100,dashed  ] plot({\x} ,{-0.51*\x } );
}
{\draw [domain = -1 : 0,samples =50, red ] plot({\x} ,{-0.4*\x -0.75*cos(\x*\x*\x* pi*3 )+0.7});}
{\draw [domain = 0 : 5, samples = 200, red  ] plot( {\x-0.2} ,{-0.08*\x+ 0.5*sin(\x*\x*\x* pi * 3)* (0.8* \x )  ) } );
}
\draw[line width=1pt,->](-0.7,0)--(5.2,0);
\draw[ line width=0.5pt , ->](0,-0.5)--(0,1.8); 
\end{tikzpicture}

\caption{The  limits inferior and superior of linear growth at $+\infty$ (or at $-\infty$)  have different signs.} \label{Fig:impossible_2}
\end{figure}

\begin{figure}[h]

\begin{tikzpicture}[scale=0.5]
{\draw [domain = -5 : 0,samples =200, red ] plot({\x} ,{-0.4*\x -0.75*cos(\x*\x*\x* pi*3 )+0.7});
}
{\draw [domain = -5 : 0,samples =200, red ] plot({-\x} ,{-0.4*\x -0.75*cos(\x*\x*\x* pi*3 )+0.7});
}

\draw[line width=1pt,->](-5,0)--(5,0);
\draw[ line width=0.5pt , ->](0,-0.5)--(0,2.5); 
\draw[ line width=0.5pt ,dashed](0,0)--(-5,1.9);
\draw[ line width=0.5pt ,dashed](0,0)--(5,1.9);
\end{tikzpicture}

\caption{At least negative (resp. positive)  linear growth at $-\infty$ (resp. $+\infty$). The symmetric case with respect to the horizontal axe is also impossible.} \label{Fig:impossible_3}
\end{figure}

We provide families of interactions fulfilling the hypotheses of this theorem in appendix \ref{Nao degenerescencia}.
Assumption $ \inf_x  E (x, x)~>~\bar E $ is a non-degeneracy hypothesis. This condition roughly indicates that 
a chain formed by atoms very close to each other cannot be strongly  minimizing,
so that Ma\~n\'e calibrated configurations have an intrinsic interspacing.

The above theorem has implications for the periodic case. As already mentioned, the existence of periodic (and therefore bounded) weak KAM solutions is well established. To the best of our knowledge, however, weak KAM solutions with linear or mixed behavior in the periodic context are not reported in the literature. 

\begin{theorem}\label{theo:PeriodicCase}
For a weakly twist and periodic interaction model $ E $ such that  $ \inf_{x} E (x, x) > \bar E $,
there are exactly three types of weak KAM solutions: those that are bounded,   those that have linear growth along the real line, and  those that are simultaneously   bounded
on a selected half-line and have linear growth along the complementary half-line.

Moreover, two weak KAM solutions belonging to the same class remain at a uniform distance from each other along the real line.
\end{theorem}

Of course a weak KAM solution with linear growth  provided by theorem~\ref{theo:PeriodicCase} is not periodic. 
Actually, for  periodic interaction, we do not known if any bounded weak KAM solution is also periodic.

In the degenerate  case, where the infimum of self-interactions, namely $\min\{E(x,x) ,  x \in \mathbb R\}$, equals the ground action, behaviors  of weak KAM solutions may be very different from those observed in theorem~\ref{theo:mainThm}, even in the periodic case. As an illustration, we propose the following result.

\begin{theorem}\label{theo:Periodic degenerate}
There is a periodic interaction $ E(x, y) $, with $ \inf_{x} E(x, x) = \bar E $, for which $ x \mapsto S(0, x) $ defines a weak KAM solution of linear growth
but all Ma\~n\'e calibrated configurations are (uniformly) bounded. 
\end{theorem}

As noted by one of the referees, more information about the preferred ordering can be proved for periodic interaction models.
\begin{proposition} \label{Proposition:ComplementPreferredOrdering}
Let $E_0$ be a periodic interaction model satisfying (H1) and (H3). 
With respect to the one-parameter family $E_\lambda(x,y) :=  E_0(x,y) - \lambda(y-x)$, the following items hold.
\begin{enumerate}
\item \label{Item:ComplementPreferredOrdering_1} $E_\lambda$ is non-degenerate for large $|\lambda|$.
\item \label{Item:ComplementPreferredOrdering_2} For positive large $\lambda$, the preferred ordering for $E_\lambda$ is increasing.
\item \label{Item:ComplementPreferredOrdering_3} For negative large $\lambda$, the preferred ordering for $E_\lambda$ is decreasing.
\end{enumerate}
If in addition $E_0$ is supposed to fulfill (H1-4), denote
\begin{gather*}
\Lambda^+ := \big\{ \lambda \in \mathbb{R} : E_\lambda \ \text{is non-degenerate and has an increasing ordering} \big\}, \\
\Lambda^- := \big\{ \lambda \in \mathbb{R} : E_\lambda \ \text{is non-degenerate and has an decreasing ordering} \big\}.
\end{gather*}
\begin{enumerate}
\addtocounter{enumi}{3}
\item \label{Item:ComplementPreferredOrdering_4} Then there exist $\lambda_- \leq \lambda_+$ such that
\begin{gather*}
\Lambda_+ = ( \lambda_+,+\infty), \qquad \Lambda_- = (-\infty ,\lambda_-),
\end{gather*}
and for every $\lambda \in [\lambda_-,\lambda_+]$ the interaction model $E_\lambda$ is degenerate.
\end{enumerate}
\end{proposition}

A classical approach to get weak KAM solution is through the study of the action of the Lax-Oleinik operator on a suitable space of functions. Related to our context, an interesting space is the one formed by continuous  functions  $u \colon \mathbb{R} \to \mathbb{R}$ with \textit{at most linear growth}, namely, fulfilling $\sup_x |u(x)|/(|x|+1)  <+\infty$. For $E$ an interaction  satisfying assumptions (H1-3) of hypotheses~\ref{assumption:interactionModel}, recall that the  
\textit{backward Lax-Oleinik operator} is  the non-linear operator acting on the space of continuous functions with at most linear growth as
\begin{equation} \label{definition:LaxOleinikOperator}
\forall \,y \in \mathbb{R}, \quad T[u](y) := \inf_{x\in\mathbb{R}} \big\{ u(x) + E(x,y) \big\}.
\end{equation}
Then weak KAM solutions are functions that satisfy $T[u] = u+ \bar{E}$. 

Our purpose in the next sections is to detail the proofs of the above results.
The rest of the paper is organized as follows. 
In section \ref{sec:suitableLaxOleinik}, from a Ma\~n\'e calibrated configuration that traverses the entire real line, 
we define a family of localized Lax-Oleinik operators preserving a suitable  sequence  of  functional spaces. 
We then show they admit additive eigenfunctions and moreover their accumulation points are indeed Lipschitz weak KAM solutions (theorem~\ref{theorem:atMostLinearWeakKAMsolution}). 
In section~\ref{secao intermediaria}, for weakly twist models $ E $ that fulfill $ \inf_x E(x,x) > \bar E $, proposition~\ref{Lemma:WeaklyTwistAdditionalProperties} ensures that calibrated configurations are always monotone and have successive jumps bounded (in a uniform way) from above as well as from below.
This provides the existence of a Ma\~n\'e calibrated configuration as required to apply  theorem~\ref{theorem:atMostLinearWeakKAMsolution}, which gives  theorem~\ref{existenciageral} that corresponds to the first statement of theorem~\ref{theo:mainThm}.
The linearly repetitive quasi-periodic case is studied in details in section~\ref{sec:LR}. 
Thanks to the specific properties of repetitions of the pattern in this quasi-periodic setting, 
we are able to identify, as $ | x - y | \to + \infty $, distinct behaviors of the Ma\~n\'e potential according to whether $ x < y $ or $ x > y $ (proposition~\ref{Proposition:FundamentalConfiguration}). 
Actually, the model introduces a preferred ordering on the real line ($\epsilon=1$ for the standard ordering, $\epsilon=-1$ for the reversed one), and a dichotomy on the type of growth of the Ma\~n\'e potential $ S(x, y) $ as  $ | x - y | \to + \infty $:
a sublinear growth of $ S(\epsilon x, \epsilon y) $ when $ x < y $ and a positive linear growth of $ S(\epsilon x, \epsilon y) $ when $ x > y $.
Such a feature is the key piece that allows the classification of all weak KAM solutions.
Hence, proposition~\ref{Proposition:WeakKAMsolutionClassificationbis}, corollary~\ref{ordem no primeiro} and proposition~\ref{distancia entre solucoes} are the results that complete the statement of theorem~\ref{theo:mainThm}.
The periodic example proving theorem \ref{theo:Periodic degenerate} is studied in proposition~\ref{proposition:periodicManePotential}, in section~\ref{periodic and linear}. 
Also in this section is the proof of proposition~\ref{Proposition:ComplementPreferredOrdering}.

We would like to thank the referees for suggesting proposition~\ref{Proposition:ComplementPreferredOrdering} and for their comments, which improved the quality of this article. 

\section{Ma\~n\'e calibration and weak KAM solutions}\label{sec:suitableLaxOleinik}

We prove in this section the existence of weak KAM solutions \eqref{Equation:DefinitionWeakKAMsolution} under the assumption there exists a Ma\~n\'e calibrated configuration unbounded at $\pm\infty$ with uniformly bounded jumps. Note that we do not require this configuration to be monotone.

During this section, we suppose that $E$ is an interaction model satisfying assumptions (H1-3) of hypotheses~\ref{assumption:interactionModel}, no additional condition is required.

We recall first some of main definitions mentioned in the introduction.

\begin{definition}  \label{definition:basicWeakKAMtools}
\
\begin{enumerate}
\item \label{Item:basicWeakKAMtools_1} We call \textit{ground action} the quantity
\[
\bar E := \lim_{n\to+\infty} \ \inf_{x_0,x_1,\ldots,x_n} \frac{1}{n} \sum_{k=0}^{n-1} E(x_k,x_{k+1})
\]
\item We call  \textit{Ma\~n\'e potential} the function defined on $\mathbb{R} \times \mathbb{R}$ as
\[
S(x,y) := \inf_{n\geq1} \ \inf_{x=x_0,x_1,\ldots,x_n=y} \ \sum_{k=0}^{n-1} \big( E(x_k,x_{k+1}) -\bar E \big).
\]
\item A subconfiguration  $(x_k)_{k=p}^q$, $p<q$, is said to be  \textit{Ma\~n\'e calibrated} if
\[
\forall \, p \leq m < n \leq q, \quad S(x_m,x_n) = \sum_{k=m}^{n-1}  \big( E(x_k,x_{k+1}) -\bar E \big).
\]
\end{enumerate}
\end{definition}

To simplify the notations we will use the convention
\[
E(x_0,x_1,\ldots,x_n) := \sum_{k=0}^{n-1}E(x_k,x_{k+1}).
\]

We observe the following simple properties.
\begin{remark} \label{remark:SimplePropertiesManePotential}
\
\begin{enumerate}
\item \label{item:SimplePropertiesManePotential_1} $\bar E =  \sup_{n\geq1}  \inf_{x_0,x_1,\ldots,x_n} \frac{1}{n} E(x_0,x_1,\ldots,x_n)$,
\item \label{item:SimplePropertiesManePotential_2} $ \inf_{x,y \in\mathbb{R}}E(x,y) \leq \bar E  \leq \inf_{x\in\mathbb{R}} E(x,x)$,
\item \label{item:SimplePropertiesManePotential_3} $\forall \,x_0 \in \mathbb{R}, \ \forall \,n\geq1, \quad \bar E \leq \inf_{x_1,\ldots,x_{n-1}} \frac{1}{n} E(x_0,x_1,\ldots,x_{n-1},x_0)$
\item\label{item:SimplePropertiesManePotential_4}  $\inf_{x,y \in \mathbb{R}} S(x,y) \leq 0 \leq \inf_{x \in\mathbb{R}} S(x,x)$,
\item\label{item:SimplePropertiesManePotential_5}  $\forall \,x,y,z \in\mathbb{R}, \ S(x,z) \leq S(x,y) + S(y,z)$,
\item\label{item:SimplePropertiesManePotential_6}  $\forall \,x,y \in \mathbb{R}, \ \bar E - E(y,x) \leq S(x,y) \leq E(x,y) - \bar E$.
\end{enumerate}
\end{remark}

Property~\eqref{item:SimplePropertiesManePotential_1} is a consequence of  Fekete's lemma and the super-additivity of $[n \mapsto \inf_{x_0,\ldots,x_n}E(x_0,\ldots,x_n)]$. 
Property~\eqref{item:SimplePropertiesManePotential_2} is obtained by bounding from above $\inf_{x_0,\ldots,x_n}E(x_0,\ldots,x_n)$ by computing the action on configurations  of the form $(x,x,\ldots,x)$. 
Property~\eqref{item:SimplePropertiesManePotential_3} follows from  Fekete's lemma and  the sub-additivity of $[n\mapsto\inf_{x_1,\ldots,x_{n-1}}  E(x_0,x_1,\ldots,x_{n-1},x_0)]$.
Property~\eqref{item:SimplePropertiesManePotential_4} is a consequence of \eqref{item:SimplePropertiesManePotential_3} for  the right hand side, and a consequence of the definition of $\bar E$ and the inequality $S(x,y) \leq E(x,y)-\bar E$ for the left hand side.
Property~\eqref{item:SimplePropertiesManePotential_5} follows by concatenation of configurations. 
Finally, property~\eqref{item:SimplePropertiesManePotential_6} follows from the inequality $S(x,y) \leq E(x,y) - \bar E$ obtained by taking a simple configuration $(x,y)$, and 
then from~\eqref{item:SimplePropertiesManePotential_5} using the second inequality of~\eqref{item:SimplePropertiesManePotential_4}.

We show in the following lemma that any weak KAM solution is Lipschitz and that any backward minimizer in the definition of the Lax-Oleinik operator \eqref{definition:LaxOleinikOperator} has a uniform bounded jump.

\begin{lemma} \label{Lemma:AprioriBoundWeakKAMsolution}
There exist constants $K,R\geq0$ (depending only on the interaction model $E$) such that for every weak KAM solution $u$
\begin{enumerate}
\item \label{Item:AprioriBoundWeakKAMsolution_2}  $u$ is Lipschitz continuous and $\text{\rm Lip}(u) \leq K$,
\item \label{Item:AprioriBoundWeakKAMsolution_1}  $\forall\, y \in \mathbb{R}, \, \, \argmin \{ u(\cdot) + E(\cdot,y) \} \subset [y-R, y+R]$.
\end{enumerate}
\end{lemma}

\begin{proof}
Let $u$ be a weak KAM solution, that is, $T[u] = u+ \bar E$, where $ T $ is the Lax-Oleinik operator associated to $ E $
defined as \eqref{definition:LaxOleinikOperator}.

\medskip
{\it Step 1.} We show an \textit{a priori} linear growth of $u$. Denote
\[
\tilde K := \sup_{|y-x| \leq 1} E(x,y) - \inf_{x,y \in \mathbb R} E(x,y).
\] 
We claim that $  |u(y)-u(x)| \leq \tilde K ( |y-x|+1) $ for all $ x, y \in \mathbb{R} $.
Indeed, either one has $|y-x| \leq 1$, and then 
\begin{equation}\label{separacao1}
u(y) - u(x) \leq E(x,y) - \bar E \leq \sup_{|y-x| \leq 1} E(x,y) - \inf_{x,y \in \mathbb R} E(x,y) \leq \tilde K, 
\end{equation}
which clearly implies $ u(y)-u(x) \leq \tilde K (|y-x|+1) $.
Or otherwise for some $n\geq 2$, $n-1 < |y-x| \leq n$. 
In this case, consider $x_k := x + \frac{k}{n}(y-x)$, $ k = 0, \ldots, n $, a sequence of points spaced apart by at most 1. 
Then from~\eqref{separacao1} $ u(x_k) - u(x_{k-1})  \leq \tilde K $, so that
\begin{equation*}
u(y)-u(x) \leq n \tilde K \leq \tilde K (|y-x|+1).
\end{equation*}

{\it Step 2.} We show item \ref{Item:AprioriBoundWeakKAMsolution_1}. 
By the superlinearity, there exists $ R \geq \tilde K $ such that $ E(x,y) > \tilde K |x - y| + \bar E + \tilde K $ whenever $ |x - y| > R $.   
Suppose $x,y \in \mathbb{R}$ fulfill $u(y)-u(x) = E(x,y) - \bar E$. 
Assume by contradiction that $|y-x| >  R$. From the first step, $ |u(y)-u(x)| \leq \tilde K (|y-x|+1) $.
However, from the choice of $ R $, we see that
$ E(x,y) - \bar E  >  \tilde K (|y-x|+1) $.
We thus obtain a contradiction and conclude that $|y-x| \leq R$.

\medskip
{\it Step 3.}  We show item \ref{Item:AprioriBoundWeakKAMsolution_2}. 
Let $y,z \in\mathbb{R}$ and $ x \in\argmin \{ u(\cdot) + E(\cdot,y)\} $.
Hence,
\begin{equation*}
u(y) = u(x) + E(x,y) - \bar E \qquad \text{and} \qquad u(z) \leq u(x) + E(x,z) -\bar E.
\end{equation*}
Either $|z-y| \ge  1$, and therefore
\begin{gather*}
u(z) - u(y) \leq \tilde K (|z-y|+1) \leq 2 \tilde K |z-y|,
\end{gather*}
or $|z-y| < 1$ so that,  as $|y-x| \leq  R$ and $|z-x| \leq  R + 1$, using the constant $C_{Lip}^E$ in (H2),
\begin{gather*}
u(z)-u(y) \leq E(x,z) - E(x,y) \leq C_{Lip}^E(R +1) |z-y|.
\end{gather*}
We obtain that $u$ is $K$-Lipschitz with  $K : =\max \{2  \tilde K,C_{Lip}^E(R +1)\}$.
\end{proof}

We highlight the key result of this section.

\begin{theorem} \label{theorem:atMostLinearWeakKAMsolution}
Assume there exists a configuration $(x_k)_{k \in \mathbb Z}$ fullfiling
\begin{itemize}
\item $(x_k)_{k\in\mathbb{Z}}$ is Ma\~n\'e calibrated,
\item it has bounded jumps, namely, $\displaystyle \sup_{k\in\mathbb{Z}} |x_{k+1}-x_k| < + \infty$,
\item it is unbounded from above and below in the sense that either
\begin{equation*} \limsup_{k\to+\infty} x_k = +\infty, \ \liminf_{k\to-\infty} x_k = -\infty, \end{equation*}
or 
\begin{equation*} \liminf_{k\to+\infty} x_k = -\infty, \ \limsup_{k\to-\infty} x_k = +\infty. \end{equation*}
\end{itemize}
Then there exists a Lipschitz weak KAM solution $ v $ such that for $ m < n $, 
the configuration $(x_k)_{k\in\mathbb{Z}}$ satisfies  
$$ v(x_n)-v(x_m) = S(x_m,x_n).$$
\end{theorem}

We first prove an \textit{a priori} linear growth of the Ma\~n\'e potential. Note that, thanks to hypotheses~(H1) and (H2),
$$ \forall \, R > 0, \qquad \sup_{|y-x| \le R} E(x,y) < +\infty. $$  

\begin{lemma} \label{lemma:aprioriLipschitzLarge}
There exists a constant $C>0$ such that
\[
\forall \,x,y \in \mathbb{R}, \quad |S(x,y)| \leq C(|y-x|+1).
\]
\end{lemma}

\begin{proof}
Define
\[
C := \sup_{|y-x| \leq 1} | E(x,y) - \bar E|.
\]
Choose $n\geq 1$ such that $n-1 \leq |y-x| < n$, and denote $t_k := x+ \frac{k}{n}(y-x)$. Then
\[
|S(x,y)| \leq \sum_{k=1}^n |S(t_{k-1},t_k)| \leq n C \leq C(|y-x|+1). \qedhere
\]
\end{proof}

\begin{proof}[Proof of theorem~\ref{theorem:atMostLinearWeakKAMsolution}]
Suppose that
\[
\limsup_{k\to+\infty} x_k = +\infty \quad \text{and} \quad \liminf_{k\to-\infty} x_k = -\infty.
\]
The other case follows from this one by introducing $ \hat E(x, y) := E(-x, -y) $ and noticing that
$ (\hat x_k)_k := (- x_k)_k $ is calibrated with respect to $ \hat S(x, y) = S(-x, -y) $.

\medskip
\textit{Step 1.} The idea of the proof is to construct (in a uniformly Lipschitz way) approximated weak KAM solution on an exhausting sequence of compact intervals $B_N := [x_{i_N},x_{j_N}]$, $i_N \to-\infty$, $j_N\to+\infty$. The difference between the standard Lax-Oleinik operator and the approximated one is that we impose on the latter a fixed boundary condition on the set of solutions outside $B_N$. 

We define inductively two sequences of indices 
\[
\cdots \leq  i_{2} \leq i_{1} < 0 < j_1 \leq j_2 \leq \cdots
\]
such that for every $N\geq 1$,
\begin{align*}
&x_{i_{N}-1} < x_0-N < x_0+ N < x_{j_N+1} \text{ and}  \\
&\forall \,i_{N} \leq k \leq j_N, \ x_k \in  [x_0-N,x_0+N].
\end{align*}
Let $B_N := [x_{i_{N}},x_{j_N}]$ and $T_N$ be the operator acting on $C^0(B_N)$ by
\[
\forall \,y \in B_N, \quad T_N[u](y) = \min_{x\in\mathbb{R}} \big[ \tilde u(x) + E(x,y) -\bar E \big],
\]
where $\tilde u$ is the extension of $u$ on $\mathbb{R}$ defined as 
\[
\forall \,x \not\in B_N, \quad \tilde u(x) = S(x_{i_N-2},x).
\]
Note that $T_N$ is well defined thanks to the superlinearity of the interaction and the sublinearity of $S$. 
We show there exists a constant $K>0$ such that for every $N$,  $T_N$ preserves the following functional space
\begin{align*}
\mathcal{H}_N := \big\{  u  \in\ & C^0(B_N) :  \forall \,i_N \leq k \leq j_N, \ u(x_k) = S(x_{i_N-2},x_k), \\
&\forall \,x \in B_N, \ u(x) \geq S(x_{i_N-2},x), \\
&\forall \,x,y \in B_N \text{ with } |x-y|<1, \ |u(y)-u(x)| \leq K|y-x|  \big\},
\end{align*}
Note that for  $u \in \mathcal{H}_N$,
\[
\forall \,x \in \mathbb{R}, \ \tilde u(x) \geq S(x_{i_N-2},x), \quad \forall \,k\in\mathbb{Z}, \ \tilde u(x_k) = S(x_{i_N-2},x_k). 
\]
To prove the invariance of $\mathcal{H}_N$ under $T_N$,  observe for every $y \in B_N $ and $x \in \mathbb{R}$,
\[
\tilde u(x) + E(x,y) - \bar E \geq S(x_{i_N-2},x) + S(x,y) \geq S(x_{i_N-2},y),
\]
and for every $i_N \leq  k \leq j_N$, thanks to the calibration of $(x_k)_{k\in\mathbb{Z}}$,
\[
\tilde u(x_{k-1}) + E(x_{k-1},x_k) - \bar E = S(x_{i_N-2},x_{k-1}) + S(x_{k-1},x_k) = S(x_{i_N-2},x_k),
\]
which implies 
\begin{gather*}
T_N[u](y) \geq S(x_{i_N-2},y), \  \forall \,y \in B_N \quad \text{and} \\
T_N[u](x_k) = S(x_{i_N-2},x_k), \  \forall \,i_N \leq k \leq j_N. 
\end{gather*}

Let $y \in B_N$. We prove that  the infimum in the definition of $T_N[u](y)$ is attained at some $x \in \mathbb{R}$ satisfying $|y-x| \leq R$ for some uniform constant $R>0$. 
Define  $\rho := \sup_{k\in\mathbb{Z}}|x_{k+1}-x_k|$. On the one hand, if $x_k \in B_N $ is chosen such that $|y-x_k| \leq \rho$, then
\begin{align*}
T_N[u](y) &\leq \tilde u(x_k) + E(x_k,y) - \bar E = S(x_{i_N-2},x_k) +  E(x_k,y)-\bar E  \\
&\leq S(x_{i_N-2},y) + S(y,x_k) + E(x_k,y)-\bar E \\
&\leq S(x_{i_N-2},y) + 2\sup_{|x-x'|\leq\rho} \big| E(x,x')-\bar E  \big|,
\end{align*}
On the over hand, by the superlinearity of the interaction, one can find $x \in \mathbb{R}$ such that
\[
T_N[u](y) = \tilde u(x) + E(x,y) - \bar E \geq S(x_{i_N-2},x) + E(x,y) - \bar E.
\]
Combining both inequalities, one obtains
\[
E(x,y) - \bar E \leq S(x,y) + D,
\]
with $D := 2 \sup_{|y-x| \leq \rho} | E(x,y)-\bar E |$. Using again the superlinearity of $E$ and the constant $C$ from lemma \ref{lemma:aprioriLipschitzLarge}, one gets for some constant $B>0$,
\begin{gather*}
(C+1)|y-x| - B \leq E(x,y) - \bar E \leq S(x,y) +D \leq C(|y-x|+1) + D, \\
|y-x| \leq R \quad\text{with}\quad R := B+C+D.
\end{gather*}

We prove that $T_N[u]$ is Lipschitz continuous. Consider $y_1,y_2 \in B_N $ with $|y_2-y_1| \leq 1$. Then there exists $x\in\mathbb{R}$, $|y_1-x| \leq R$, such that
\begin{gather*}
T_N[u](y_1) = \tilde u(x) + E(x,y_1) \quad \text{and} \quad  T_N[u](y_2) \leq  \tilde u(x) + E(x,y_2).
\end{gather*}
Using the constant $C_\text{Lip}^E(R)$ as in~(H2) and denoting $K:= C_\text{Lip}^E(R+1)$, one obtains
\[
|T_N[u](y_2) - T_N[u](y_1) | \leq K |y_2 -y_1|.
\]
In conclusion, $\mathcal{H}_N$ is a compact convex subset of $C^0(B_N)$ for the uniform topology. 
The non linear operator $T_N : \mathcal{H}_N \to \mathcal{H}_N$ is $1$-Lipschitz. 
By Schauder-Tychonoff theorem, $T_N$ admits a fixed point $u_N$. 

\textit{Step 2.} Define 
\[
v_N(y) := u_N(y) -u_N(x_0), \quad \forall \,y \in B_N.
\]
Then, for $N$ sufficiently large, $x_{i_N} < x_0-R < x_0 + R < x_{j_N}$,
\begin{itemize}
\item $v_N(x_0)=0$,
\item $\forall \,x,y \in B_N  \text{ with } |x-y|<1, \ |v_N(y) - v_N(x)| \leq K|y-x|$,
\item $\forall \,x ,y\in B_N, \ v_N(y) \leq v_N(x) + E(x,y) - \bar E$,
\item $\forall \,y \in [x_{i_N}+R,x_{j_N}-R], \ \exists \, x \in B_N$  such that $|y-x|\leq R$ and  $ v_N(y) = v_N(x) + E(x,y) - \bar E$,
\item $\forall \,i_N \leq k < l \leq j_N, \ v_N(x_l) - v_N(x_k) = S(x_k,x_l)$.
\end{itemize}
By using a diagonal procedure of extraction, there exists a subsequence of $(v_N)_{N}$ that converges uniformly on any compact interval to a $K$-Lipschitz function $v:\mathbb{R} \to \mathbb{R}$  that is a weak KAM solution calibrating $(x_k)_{k\in\mathbb{Z}}$.
\end{proof}

\section{Non-degenerate and weakly twist models}\label{secao intermediaria}

The main result of this section guarantees that, for weakly twist models (\textit{i.e.}, interactions fulfilling all the assumptions (H1-4) of hypotheses~\ref{assumption:interactionModel}) that satisfy the non-degenerate condition $ \inf_x E(x,x) > \bar E $, there are always weak KAM solutions. 
In order to apply theorem~\ref{theorem:atMostLinearWeakKAMsolution}, we prove in lemma~\ref{concordancia fundamentais e calibradas} the existence of increasing as well as decreasing Ma\~n\'e calibrated configurations with bounded jumps and unbounded from above and below. 
Actually we improve a result obtained in \cite{GaribaldiPetiteThieullen2017} for which the environment is supposed to be a 
quasi-crystal -- in particular, it possesses a uniquely ergodic hull and the interaction is pattern equivariant.
On the contrary lemma~\ref{concordancia fundamentais e calibradas} does not require any particular assumption
on the structure of an underlying substrate.

We first gather results that have been proved in \cite{GaribaldiPetiteThieullen2017}.

\begin{lemma} \label{Lemma:ManePotential}
Let $E$ be a weakly twist interaction. Then
\begin{enumerate}
\item $\forall\, x<y, \ S(x,y) = \inf_{x=x_0 < x_1 < \cdots < x_n=y} \big\{ E(x_0,\ldots,x_n) -n \bar E \big\}$, \label{item 1.Lemma:ManePotential}
\item $\forall\, x > y, \ S(x,y) = \inf_{x=x_0 > x_1 > \cdots > x_n=y} \big\{ E(x_0,\ldots,x_n) -n \bar E \big\}$, 
\item $\forall\, x \in\mathbb{R}, \ S(x,x) = E(x,x) - \bar E$.
\label{item 3.Lemma:ManePotential}
\end{enumerate}
Moreover, if the interaction is pattern equivariant with respect to a quasi-periodic set, then the Ma\~n\'e potential is also pattern equivariant. 
\end{lemma}

\begin{proof}
See proposition 24 in \cite{GaribaldiPetiteThieullen2017}.
\end{proof}
    
We assume from now on that $\inf_{x \in\mathbb{R}} E(x,x) > \bar E$. 
We choose once for all $\eta_0>0$ such that 
\begin{equation}
\forall\, x, y \in\mathbb{R}, \ |y-x| < \eta_0 \ \Rightarrow \  E(x,y) -\bar E > \eta_0. \label{definicao de eta}
\end{equation}

\begin{lemma}\label{lematecnico}
Assume that $E$ is a weakly twist interaction model satisfying $\inf_{x} E(x,x) > \bar{E}$. 
Then, there exist constants $A_0,B_0 > 0$ such that for any  $(x_0, \ldots, x_n) \in \mathbb{R}^{n+1}$, with $x_0=x$ and $x_n=y$, \begin{equation}
n \leq A_0 |y-x| + B_0\sum_{k=1}^n  \big( E(x_{k-1},x_k) - \bar E \big), \label{Equation:WeaklyTwistAdditionalProperties}
\end{equation}
where  $A_0 = (\bar E - \inf_{x,y \in \mathbb{R}}E(x,y) +\eta_0)/\eta_0^2$ and $B_0 = 1/\eta_0$.
\end{lemma}

\begin{proof}
If $(x_0, \ldots, x_n)$ is not monotone, then by lemma 23 in \cite{GaribaldiPetiteThieullen2017} there exists a subset of distinct indices $\{ i_0, \ldots, i_\sigma\}$ of $\{0, \ldots, n\}$, with $i_0=0$, $i_\sigma=n$, such that $(x_{i_0}, \ldots, x_{i_\sigma})$ is strictly monotone, and
\begin{equation}\label{desigualdade_twist}
E(x_0, \ldots, x_n) - n \bar E > E(x_{i_0}, \ldots, x_{i_\sigma}) -\sigma \bar E + \sum_{j \not\in \{ i_0, \ldots, i_\sigma\}} \big( E(x_j,x_j) - \bar E \big).
\end{equation}
If $(x_0, \ldots, x_n)$ is monotone, we choose $\sigma=n$ and $i_k = k$ for all $k $. In both cases, we obtain
\begin{gather*}
E(x_0, \ldots, x_n) - n \bar E \geq E(x_{i_0}, \ldots, x_{i_\sigma}) -\sigma \bar E + (n-\sigma) \eta_0.
\end{gather*}
We now consider the set of indices $I \subseteq \{0, \ldots, \sigma-1\}$ such that $k \in I$ if and only if  $|x_{i_k}- x_{i_{k+1}}| \geq \eta_0$. 
If $k \in I$, we use an \textit{a priori} lower bound
\[
E(x_{i_k}, x_{i_{k+1}}) - \bar E \geq  E_{min} - \bar E
\]
where $E_{min} = \inf_{x,y \in\mathbb{R}} E(x,y)$. If $k \not\in I$, the definition of $\eta_0$ gives
\[
E(x_{i_k}, x_{i_{k+1}}) - \bar E \geq \eta_0.
\]
Hence, we have
$$ E(x_{i_0}, \ldots, x_{i_\sigma}) - \sigma \bar E \geq |I| ( E_{min} -\bar E) + (\sigma-|I|) \eta_0. $$
Combining the estimates above, we obtain 
$$
E(x_0, \ldots, x_n) - n \bar E \geq |I| ( E_{min} -\bar E - \eta_0) + n \eta_0. 
$$
By monotonicity of $(x_{i_0}, \ldots , x_{i_\sigma})$, clearly $|I| \leq |x_{i_\sigma} - x_{i_0}|/\eta_0 = |x_{n} - x_{0}|/\eta_0$, so that
\begin{gather*}
n \leq \frac{\bar E - E_{min} +\eta_0}{\eta_0^2} |y-x| + \frac{1}{\eta_0} \sum_{k=1}^n \big( E(x_{k-1},x_k) - \bar E \big).
\end{gather*}
\end{proof}

We show in the following lemma that the infimum in the definition in $S(x,y)$ is actually a minimum, and that the number of points realizing the minimum is bounded from above by a quantity proportional to  $|y-x|$. 
Besides, we prove that the successive jumps of Ma\~n\'e calibrated configurations are uniformly bounded from above and from below.

\begin{proposition} \label{Lemma:WeaklyTwistAdditionalProperties}
Suppose that $E$ is a weakly twist interaction model fulfilling $\inf_{x} E(x,x) > \bar{E}$. 
\begin{enumerate}
\item  \label{Item:WeaklyTwistAdditionalProperties_2} 
For every $x \not= y$, there are an integer $n\geq1$ and a strictly monotone configuration $ (x_0, \ldots, x_n) $, with $ x_0 = x $ and $ x_n = y $, fulfilling  
\begin{equation*}
S(x,y) = \sum_{k=1}^n E(x_{k-1},x_k) - n\bar E.
\end{equation*}

\item \label{Item:WeaklyTwistAdditionalProperties_1}   
There exist constants $ A > 0 $ and $ B \ge 0 $ such that, for every pair of points $ x, y \in \mathbb R $,
if $(x_0 \ldots, x_n) \in \mathbb{R}^{n+1}$ satisfies $x_0=x$, $x_n=y$, and $S(x,y) = \sum_{k=1}^n E(x_{k-1},x_k) - n\bar E$, then
\begin{align*}
& S(x_p, x_q) = \sum_{k=p+1}^q E(x_{k-1}, x_k) - (p-q)\bar E, \quad \forall \, 0 \le p < q \le n, \\
& (x_0, \ldots,x_n) \ \text{is monotone} \quad \text{and}  \quad \ n \leq A |y-x| + B.
\end{align*}
Besides, for $ n \ge 2 $, $ (x_0, \ldots, x_n) $ is strictly monotone.

\item  \label{Item:WeaklyTwistAdditionalProperties_3}  
The Ma\~n\'e potential admits a negative sublinear lower bound in the following sense:
\[
\forall\, \alpha>0, \ \exists\, \beta\geq0, \  \forall\, x, y \in \mathbb{R}, \ S(x,y) \geq -\alpha |y-x| - \beta.
\]
\item  \label{Item:WeaklyTwistAdditionalProperties_4}  
There exist constants $0 < r < R$ such that  every Ma\~n\'e calibrated subconfiguration  $ (x_p, \ldots, x_q) $ with $q-p\geq2$ is strictly monotone and satisfies
\begin{gather*}
\forall\, p \le k < q, \quad \ r \leq  |x_{k+1} - x_k| \leq R.
\end{gather*}
\end{enumerate}
\end{proposition}

\begin{proof}
We assume in all items $x \le y$. The other case $x \ge y$ is similar.  
We start by proving item \ref{Item:WeaklyTwistAdditionalProperties_1}.

\medskip
{\it Item \ref{Item:WeaklyTwistAdditionalProperties_1}.} 
Let $(x_0,\ldots, x_n)$ be a  configuration satisfying 
\[
S(x,y) = \sum_{k=1}^n E(x_{k-1},x_k) - n\bar E.
\] 
Obviously for $ 0 \le p < q \le n $, 
\begin{multline*}
S(x_0, x_q) +  \sum_{k=p+1}^q E(x_{k-1}, x_k) - (p-q) \bar E + S(x_q, x_n) \le \\
 \le \sum_{k=1}^n E(x_{k-1},x_k) - n\bar E = S(x_0, x_n)  \le S(x_0, x_q) + S(x_p, x_q) + S(x_q, x_n)
\end{multline*}
implies that $  S(x_p, x_q) =  \sum_{k=p+1}^q E(x_{k-1}, x_k) - (p-q) \bar E $. 
Inequality~\eqref{desigualdade_twist} shows that $(x_0, \ldots, x_n)$ must be monotone, since otherwise one could decrease strictly the Ma\~n\'e potential
\begin{gather*}
S(x,y) = E(x_0, \ldots, x_n) - n \bar E > E(x_{i_0}, \ldots, x_{i_\sigma}) - \sigma \bar E \geq S(x,y).
\end{gather*}
For $n \ge 2 $, the configuration is actually strictly monotone,
since otherwise there would exist $1\le j \le n$ such that $x_{j-1}=x_j$ and we would obtain the same contradiction as above
\begin{align*}
S(x,y) & = \Big[ \sum_{k \neq j} E(x_{k-1},x_k) - (n-1)\bar E \Big] + \big[ E(x_{j-1},x_j) - \bar E \big] \\
& \geq S(x,y) + \inf_{x\in\mathbb{R}} E(x,x) - \bar E.
\end{align*}
The estimate \eqref{Equation:WeaklyTwistAdditionalProperties} thus implies
$ n \leq A_0 |y-x| + B_0 S(x,y) $.
We conclude using the a priori sublinearity estimate for the Ma\~n\'e potential in lemma~\ref{lemma:aprioriLipschitzLarge}, so that
$ n \leq A|y-x| + B $, 
with $A=A_0+B_0 C$ and $B=B_0 C$.

\medskip
{\it Item \ref{Item:WeaklyTwistAdditionalProperties_2}.}
Let us consider a sequence $S_\ell > S(x,y)$ converging to $S(x,y)$. 
Item~\ref{item 1.Lemma:ManePotential} of lemma~\ref{Lemma:ManePotential} shows there exists a strictly increasing configuration $(x^{\ell}_0, \ldots, x^{\ell}_{n_\ell})$, with $ x^{\ell}_0 = x $ and $ x^{\ell}_n = y $, such that 
\[
S_\ell > \sum_{k=1}^{n_\ell} \big( E(x_{k-1},x_k) - n_\ell\bar E \big).
\] 
The estimate~\eqref{Equation:WeaklyTwistAdditionalProperties} implies
$ n_\ell \leq A_0 |y-x| + B_0 S_\ell $.
As $S_\ell \to S(x,y)$, we may assume $n_\ell=n$ is constant. 
We then extract a subsequence of $(x^{\ell}_0, \ldots, x^{\ell}_{n})$ converging to some $(x_0, \ldots, x_{n})$ satisfying
\begin{gather*}
S(x,y) \geq \sum_{k=1}^{n} \big( E(x_{k-1},x_k) - n\bar E \big) \geq S(x,y).
\end{gather*}
The previous item shows that $(x_0, \ldots, x_n)$ is strictly monotone.

\medskip
{\it Item \ref{Item:WeaklyTwistAdditionalProperties_3}.} 
Let $\alpha>0$ and $\alpha'= \alpha / A$, where $ A $ is the constant obtained in the first item.
Thanks to item~\ref{Item:basicWeakKAMtools_1} of definition~\ref{definition:basicWeakKAMtools}, there exists $\beta' \geq0$ such that
\[
\forall \, n \ge 1, \ \ \forall\, (x_0, \ldots, x_n), \quad E(x_0, \ldots, x_n) - n\bar E \geq -\alpha' n -\beta'.
\]
Items~\ref{Item:WeaklyTwistAdditionalProperties_2} and~\ref{Item:WeaklyTwistAdditionalProperties_1} of the present proposition ensure that there is a particular configuration $(x_0, \ldots, x_n)$ such that
\begin{align*}
S(x,y) &= E(x_0, \ldots, x_n) - n \bar E \geq -\alpha' n - \beta' \\
&\geq -\alpha' (A|y-x|+B) -\beta' = -\alpha |y-x| -\beta,
\end{align*}
with $\beta := \alpha' B + \beta'$.

\medskip
{\it Item \ref{Item:WeaklyTwistAdditionalProperties_4}.} Let $(x_p, \ldots, x_q) $ be a Ma\~n\'e calibrated subconfiguration. 
It is strictly monotone for  $ q - p \ge 2 $ as a consequence of item~\ref{Item:WeaklyTwistAdditionalProperties_1}.
From lemma~\ref{lemma:aprioriLipschitzLarge}, we have 
$$ S(x_k, x_{k+1}) \le C ( | x_{k+1} -  x_k | +1 ) $$ 
for some constant $ C $.
From the superlinearity of the interaction, there exists a constant $ B > 0 $ such that 
$$ (C+1)  | x_{k+1} -  x_k |  - B \le E(x_k, x_{k+1}) - \bar E = S(x_k, x_{k+1}). $$ 
Therefore, $ | x_{k+1} -  x_k | \le B + C := R $.
With respect to the lower bound, let first $ \eta_0 > 0 $ be defined as in~\eqref{definicao de eta}.
Note then that for $ p \le k < q - 1 $,
\begin{align}
S(x_k, x_{k+1}) & = S(x_k, x_{k+2}) - S(x_{k+1}, x_{k+2}) \nonumber \\
& \le E(x_k, x_{k+2}) - E(x_{k+1}, x_{k+2}) \nonumber \\
& \le C_\text{Lip}^E (2R) \, | x_{k+1} - x_k |. \label{Lipschitz calibracao}
\end{align}
We claim that $ | x_{k+1} - x_k | > \eta_0 / \big(C_\text{Lip}^E (2R) + 1\big) =: r $. Indeed, otherwise by the very definition of $ \eta_0 $
we would have $ \eta_0 <  E(x_k, x_{k+1}) - \bar E = S(x_k, x_{k+1}) $, but~\eqref{Lipschitz calibracao} shows that 
$ S(x_k, x_{k+1}) \le  C_\text{Lip}^E (2R) \, \eta_0 / \big(C_\text{Lip}^E (2R) + 1\big) < \eta_0 $, and we would reach a contradiction.
The equality $ S(x_{q-1}, x_q) = S(x_{q-2}, x_q) - S(x_{q-2}, x_{q-1}) $ allows to discuss the case of the last index in a similar way.
\end{proof}

The regularity of the Ma\~n\'e potential is an immediate consequence of the previous proposition.

\begin{corollary}\label{potencial Lipschitz}
For a weakly twist interaction $ E $ fulfilling $ \inf_x E(x,x) > \bar E $, the Ma\~n\'e potential is Lipschitz continuous. 
\end{corollary}

\begin{proof}
Let $ I, J \subset \mathbb R $ be both open intervals of length 1. It is enough to argue that $ S |_{I \times J} $ is Lipschitz.
Items~\ref{Item:WeaklyTwistAdditionalProperties_2} and~\ref{Item:WeaklyTwistAdditionalProperties_1} of proposition~\ref{Lemma:WeaklyTwistAdditionalProperties} guarantees there exists $ N = N(I,J) > 0 $ such that
$$ S(a, b) = \min_{1\le n \le N} \min_{(a=x_0, \ldots, x_n=b)} \big[ E(x_0, \ldots, x_n) - n \bar E \big], $$
for all $ a \in I $ and $ b \in J $.  Therefore, given $ x, \hat x \in I $ and $ y, \hat y \in J $, items~\ref{Item:WeaklyTwistAdditionalProperties_2} and~\ref{Item:WeaklyTwistAdditionalProperties_4} of proposition~\ref{Lemma:WeaklyTwistAdditionalProperties} provide the estimate
\begin{align*}
| S(x, y) - S(\hat x, \hat y)| \le  \max &  
\Big\{ \max_{\substack{|x-a|\le R \\ |y-b|\le R}} \left|E(x, a) - E(\hat x,a) + E(b, y) - E(b, \hat y)\right|, \\
& \max_{\substack{|\hat x- \hat a|\le R \\ |\hat y- \hat b|\le R}} \big|E(x, \hat a) - E(\hat x, \hat a) + E(\hat b, y) - E(\hat b,\hat y)\big| \Big\}.
\end{align*}
Since $ E $ is locally uniformly Lipschitz, we conclude that
$$ | S(x, y) - S(\hat x, \hat y)| \le C_\text{Lip}^E(R + 1) \, \big(|x -\hat x| + |y - \hat y| \big). $$ 
\end{proof}

The existence of Ma\~n\'e calibrated configurations such as those required among the hypotheses of theorem~\ref{theorem:atMostLinearWeakKAMsolution} was actually proved in~\cite{GaribaldiPetiteThieullen2017} by adopting a viewpoint focused on minimizing Mather measures, a similar strategy to the one inaugurated by Mather~\cite{Mather1991}. We actually do not want to discuss minimizing measures in the present paper and prove the existence of Ma\~n\'e calibrated configurations  in a more direct way.

\begin{lemma}\label{concordancia fundamentais e calibradas}
Assume that $E$ is  a weakly twist interaction model  such that $ \inf_x E(x, x)~>~\bar E $. Then $ E $ admits increasing as well as decreasing
Ma\~n\'e calibrated configurations which have bounded jumps and are unbounded from above and below.
\end{lemma}

\begin{proof}
We make use of $ R, r > 0 $, the constants that bound the successive jumps of calibrated configurations according to item~\ref{Item:WeaklyTwistAdditionalProperties_4} of proposition~\ref{Lemma:WeaklyTwistAdditionalProperties}. 
In particular, given $ A \in \mathbb R $ with $ |A| > R $, from proposition~\ref{Lemma:WeaklyTwistAdditionalProperties}
we consider a subconfiguration $ (x_p^A, \ldots, x_q^A) $, with $ p < 0 < q $, such that 
\begin{align*}
& |x_0^A| \le R,  \\
& r \le |x_{k+1}^A - x_k^A| \le R \quad \forall \, k,  \\
& S(-A, A) = E_{q-p}(x_p^A, \ldots, x_q^A) - (q-p) \bar E.
\end{align*}
Note that $ (x_p^A, \ldots, x_q^A) $ is increasing for $ A > 0 $ and decreasing for $ A < 0 $.
By denoting $ x_k^A := x_p^A = -A $ for all $ k \le p $ and $ x_k^A = x_q^A = A $ for all $ k \ge q $, 
we have a con\-fig\-u\-ra\-tion $ (x_k^A)_{k \in \mathbb Z} $ that belongs to the compact set 
$ \prod_{i \in \mathbb Z} \big[ -(|i| +1) R, (|i|+1) R \big] $. 
Hence, as either $ A \to +\infty $ or $ A \to -\infty $,
we are able to obtain an ac\-cu\-mu\-la\-tion point $ (y_k)_{k \in \mathbb Z} $. 
By taking into account a suitable subfamily, we may suppose that $ (y_k) $ is the limit of $ (x_k^A) $.
Obviously $ (y_k) $ fulfills for all $ k \in \mathbb Z $, $ r \le |y_{k+1} - y_k| \le R $. 
It only remains to show that it is a calibrated configuration. However, thanks to the continuity of the Ma\~n\'e potential, from item~\ref{Item:WeaklyTwistAdditionalProperties_1} of proposition~\ref{Lemma:WeaklyTwistAdditionalProperties}, it follows 
for all $ i < j $,
\begin{align*}
S(y_i, y_j) & = \lim_A S(x_i^A, x_j^A) \\ 
& = \lim_A \big[E_{j-i}(x_i^A, \ldots, x_j^A) - (j-i) \bar E\big]\\
& = E_{j-i}(y_i, \ldots, y_j) - (j-i) \bar E.
\end{align*}
\end{proof}

From the previous lemma,  one has the following consequence of theorem~\ref{theorem:atMostLinearWeakKAMsolution}.

\begin{theorem}\label{existenciageral}
Let $ E (x, y) $ be a weakly twist interaction model satisfying  $ \inf_{x} E (x, x) > \bar E $. 
Then there exist a weak KAM solution $u$ and a $u$-calibrated (and thus Ma\~n\'e calibrated) configuration $ (x_k)_{k \in \mathbb Z} $. 
There exist constants $K>0$ and $0 < r < R$ such that  every weak KAM solution $u$ is Lipschitz with $\text{\rm Lip}(u) \leq K$, and every Ma\~n\'e calibrated configuration $(x_k)_{k\in\mathbb{Z}}$ is strictly monotone and satisfies
\begin{gather*}
\forall\, k \in \mathbb{Z} , \ r \leq  |x_{k+1}-x_k| \leq R.
\end{gather*}
\end{theorem}

The results in the previous theorem and lemma~\ref{Lemma:AprioriBoundWeakKAMsolution} constitute the first two statements of theorem~\ref{theo:mainThm}.

\section{Linearly repetitive quasi-periodic sets}\label{sec:LR}

Our main goal is to show that, in the context of linearly repetitive quasi-periodic sets (see definition~\ref{RepetitiveDeloneSet} below), all weak KAM solutions are of one three types, which may be described according to their kind of growth (linear \textit{versus} sublinear), to the ordering of their calibration, or to the existence or not of calibrated configuration traversing the real line. 
Proposition~\ref{Proposition:WeakKAMsolutionClassificationbis} and corollary~\ref{ordem no primeiro} gather the core of this classification.
Along with proposition~\ref{distancia entre solucoes}, they complete the statement of theorem~\ref{theo:mainThm}.
One of the essential ingredients to reach the classification is the notion of \textit{fundamental configuration}, that is, a finite configuration that performs the minimum sum of a certain fixed number of interactions. 
 In the linearly repetitive framework, with the hypothesis of non-degeneracy $\inf_{x \in\mathbb{R}} E(x,x) > \bar E$, all these fundamental configurations are shown to be ordered in the same way as long as a large enough number of interactions is considered.
This define a preferred ordering. 
The linear repetitivity and the non-degeneracy hypotheses allow  us to show that the Ma\~n\'e potential has a sublinear growth according to this preferred ordering (see proposition~\ref{Proposition:FundamentalConfiguration}).
In fact, the understanding of the behavior of the  Ma\~n\'e potential against the ordering introduced by these sufficiently large fundamental configurations is the key element for the study of the possible types of weak KAM solutions.

We begin by reestablishing repetitivity, now in more quantitative terms.

\begin{definition} \label{RepetitiveDeloneSet}
 A  discrete set $ \omega  \subset \mathbb {R}$,
is said  {\it repetitive} if for every $R>0$, there exists $M(R)>0$ such 
that, for any open interval $ J $ of length at least $M(R)$ and 
any pattern $ \texttt{P} $ of diameter at most $ R $, 
there is $ t \in \mathbb R $ for which $ \texttt{P} + t $ is a pattern of $ \omega \cap J $.
Besides,  whenever there are positive constants $ A $ and $ B $ such that $ M(R) \le A R + B $ for all $ R $, $ \omega $ is said to be \textit{linearly repetitive}.
\end{definition}

First note that the repetitivity  implies that the quasi-periodic set  $\omega$ is \textit{relatively dense}, {\it i.e.} there is no arbitrary large gap between consecutive elements. More quantitatively, there is a constant $R_\omega>0$ such that 
\begin{align}\label{Rzero}
\omega \cap I \neq \emptyset \quad \textrm{ for any interval } I   \textrm{ of length greater than } R_\omega. 
\end{align}

We assume from now on that the interaction is pattern equivariant with respect to a quasi-periodic set.

\begin{definition} \label{Definition:FundamentalConfiguration}
For a given interaction $E$, a {\it fundamental configuration of size $n\geq1$} is a finite sequence $(z_0, \ldots , z_n)$ such that 
\[
E(z_0,\ldots, z_n) = \min_{x_0,\ldots,x_n \in\mathbb{R}}E(x_0,\ldots,x_n).
\]
We denote by $\Gamma_n (E) \subset \mathbb{R}^{n+1}$ the set of fundamental configurations of size $n$.
\end{definition}

The above minimum exists because of the superlinearity of $E$. Moreover, by definition of $\bar E$, for any sequence of fundamental configurations $(z_0^n,\ldots,z_n^n)$
\begin{gather*}
E(z_0^n,\ldots, z_n^n)  \leq n\bar E, \\
\lim_{n\to+\infty} \frac{1}{n} E(z_0^n,\ldots, z_n^n) = \sup_{n\to+\infty} \frac{1}{n} E(z_0^n,\ldots, z_n^n) = \bar E.
\end{gather*}

We first recall the notion of minimizing configurations which is weaker than Ma\~n\'e calibration. In particular, there is no need to introduce the ground action $\bar E$.

\begin{definition}
Let $E$ be an interaction model, $n\geq1$, and $(x_0, \ldots, x_n)$ be a sequence of (possibly unordered) points of $\mathbb{R}$. The finite configuration $(x_0, \ldots, x_n)$ is said to be {\it minimizing} if, for every configuration $(y_0, \ldots, y_n)$ of the same cardinality and same extremities, $x_0=y_0$ and $x_n=y_n$, one has
\begin{gather*}
 E(x_0,x_1,\ldots, x_n) \leq E(y_0,y_1,\ldots, y_n).
\end{gather*}
\end{definition}

We show in the next lemma that minimizing configurations are strictly monotone provided that their endpoints are sufficiently far apart from each other.
Recall that $ \varsigma_0 $ is the constant that characterizes the pattern equivariance of  $ E (x, y) $ (see definition~\ref{invariancia de motivos}).

\begin{lemma} \label{Lemma:AprioriStrictMonotonicity}
Let $E$ be a weakly twist interaction model that is pattern equivariant with respect to a quasi-periodic set $ \omega $. 
Then there exists a constant $L>0$ such that every minimizing configuration $(x_0,\ldots, x_n)$ satisfying $|x_n-x_0|\geq L$  is strictly monotone.
\end{lemma}

\begin{proof}
The proof is by contradiction. To simplify, assume $x_0<x_n$.  
In the case the configuration is not monotone, by lemma 23 in \cite{GaribaldiPetiteThieullen2017} there are 
indices $i_0=0 < i_1 < \cdots < i_\sigma=n$ such that $(x_{i_0}, x_{i_1}, \ldots, x_{i_\sigma})$ is strictly monotone and
\begin{gather}
\sum_{k=0}^{n-1} E(x_k,x_{k+1}) \geq \sum_{k=0}^{\sigma-1} E(x_{i_k}, x_{i_{k+1}}) + \sum_{j \notin \{i_0, \ldots, i_\sigma\}} E(x_j,x_j). \label{Equation:AprioriStrictMonotonicity_1}
\end{gather}
In the case the configuration is monotone but not strictly monotone, we obtain the existence of $\sigma$ and $(i_0,\ldots, i_\sigma)$ as above but with an equality in~\eqref{Equation:AprioriStrictMonotonicity_1} instead of an inequality. We now use the pattern equivariance of $E$ to transport the points $x_j$, $j \notin \{i_0, \ldots, i_\sigma\}$, to new locations $\tilde x_j$  inside the interval $(x_0,x_n)$. 
Let $\mathcal{P}$ be the set of all patterns of the form 
$ \texttt{P}_x := (x-R_\omega-\varsigma_0,x+R_\omega+\varsigma_0)\cap\omega$ where $x$ is any point. By repetitivity, there exists $L>0$ such that any interval of length $L$ contains a translate of any pattern in $\mathcal{P}$. Then, for $ x_n-x_0 \geq L$, there are  $\tilde x_j$, $j\notin \{i_0,\ldots,i_\sigma\}$, such that 
\begin{equation*}
\texttt{P}_{\tilde x_j} \subset (x_0,x_n) \qquad \text{ and }  \qquad \texttt{P}_{\tilde x_j} - \tilde x_j = \texttt{P}_{x_j} - x_j.
\end{equation*}
By pattern equivariance $E(x_j,x_j) = E(\tilde x_j,\tilde x_j)$. 

Whenever $\tilde x_j \in (x_{i_s}, x_{i_{s+1}})$ for some $ j $ and $ s $, by Aubry crossing Lemma
\begin{gather*}
E(x_{i_s}, x_{i_{s+1}}) + E(\tilde x_j,\tilde x_j) >E(x_{i_s}, \tilde x_j) + E(\tilde x_j, x_{i_{s+1}}).
\end{gather*}
We may re-index the new set $\{x_{i_0}, \ldots ,x_{i_\sigma}\} \cup \{\tilde x_j\}$ as $\{\tilde x_{i_0}, \ldots , \tilde x_{i_{\sigma+1}}\}$ and again apply again Aubry crossing Lemma  to  other points $\tilde x_k$ distinct from $\{ \tilde x_{i_0}, \ldots , \tilde x_{i_{\sigma+1}}\}$. 
We finally obtain a new monotone sequence $(\tilde x_0, \ldots, \tilde x_n)$, with $ \tilde x_0 = x_0 $ and $ \tilde x_n = x_n $, satisfying 
\begin{gather*}
\sum_{k=0}^{n-1} E(x_k,x_{k+1})  >\sum_{k=0}^{n-1} E(\tilde x_k,\tilde x_{k+1}).
\end{gather*}
The strict inequality shows that $(x_0,\ldots,x_n)$ is not minimizing. We have obtained a contradiction.

We are led to consider the situation in which all the new points $\tilde x_j$ belong to $\{x_{i_1}, \ldots, x_{i_{\sigma-1}}\}$. 
By re-indexing, one obtains a monotone but not strictly monotone configuration 
$ x_0=\tilde x_0< \tilde x_1 \leq \ldots \leq  \tilde x_{n-1}< \tilde x_n=x_n $ fulfilling
\begin{equation*} 
\sum_{k=0}^{n-1} E(x_k,x_{k+1})  \geq \sum_{k=0}^{n-1} E(\tilde x_k,\tilde x_{k+1}). 
\end{equation*}
Proposition 25 of \cite{GaribaldiPetiteThieullen2017}  implies that $(\tilde x_0, \ldots, \tilde x_n)$ is not minimizing, and therefore $(x_0, \ldots x_n)$ is not minimizing. We have reached again a contradiction. 
\end{proof}

Note that $ (z_0, \ldots, z_n) \in \Gamma_n(E) $ implies $ E(z_0, \ldots, z_n) = \min_x T^n[0](x) $, where $T$ stands for the Lax-Oleinik operator introduced in \eqref{definition:LaxOleinikOperator}. In the next lemma, we guarantee that any configuration $ (y_{-n}, \ldots, y_0) $, with endpoints sufficiently apart from each other, such that $ E(y_{-n}, \ldots, y_0) = T^n[0](y_0) $ has (uniformly) bounded jumps.

\begin{lemma} \label{Lemma:BoundedJumpBackwardOptimal}
Let $E$ be a weakly twist interaction model that is pattern equivariant with respect to a quasi-periodic set $ \omega $. 
Then there exist constants $L > R > 0$  such that, for every configuration $(y_{-n}, \ldots, y_0)$ satisfying 
$|y_{-n}-y_0| \geq L$ and $ E(y_{-n}, \ldots, y_0) = T^n[0](y_0) $,
\begin{gather}
 |y_{-k+1} - y_{-k} | \leq R, \quad \forall\, 1 \le k \le n. \label{Equation:BoundedJumpBackwardOptimal_01}
\end{gather}
\end{lemma}

\begin{proof}
Although we deal with a more general context, the proof is very similar to the one of proposition 39 in \cite{GaribaldiPetiteThieullen2017}.  

\medskip
\noindent {\it Part 1.} We  prove first an intermediate result: there exists a constant $R'>0$ such that
\begin{gather*}
|y_{-n+1} - y_{-n} | \leq R', \qquad |y_{-n+2} - y_{-n+1} | \leq R'.
\end{gather*}
For the first estimate, denoting $E^{sup} := \sup_{x \in \mathbb{R}}E(x,x)$, since 
$ T^n[0](y_0) \leq E(y_{-n+1}, y_{-n+1}, y_{-n+2}, \ldots, y_0) $, one has
\begin{equation*}
E(y_{-n}, y_{-n+1}) - E^{sup} \leq E(y_{-n+1}, y_{-n+1}) - E^{sup} \leq 0.
\end{equation*}
With respect to the second estimate, introducing $E^{inf} := \inf_{x,y \in \mathbb{R}} E(x,y)$, note that
$ T^n[0](y_0) \leq E(y_{-n+2}, y_{-n+2}, y_{-n+2}, y_{-n+3}, \ldots, y_0) $ obviously implies 
$ E(y_{-n}, y_{-n+1}, y_{-n+2}) \leq E(y_{-n-2}, y_{-n+2}, y_{-n+2}) $, so that
\begin{equation*}
E(y_{-n+1}, y_{-n+2}) - E^{sup} \leq  E^{sup} - E^{inf}.
\end{equation*}
Superlinearity ensures there is $ R' > E^{sup} - E^{inf} $ such that
$$|x-y| > R' \ \Rightarrow  \  E(x,y) - E^{sup} > |x-y|. $$ 
Therefore, we necessarily have $ |y_{-n+1} - y_{-n} | \leq R' $ and $  |y_{-n+2} - y_{-n+1} | \leq R' $.

\medskip
\noindent {\it Part 2.} Lemma \ref{Lemma:AprioriStrictMonotonicity} shows that $(y_{-n}, \ldots, y_0)$ is strictly monotone.  
To fix ideas, suppose that $(y_{-n}, \ldots, y_0)$ is increasing.
Let $I $ denote the interval $ (y_{-n} - R_\omega - \varsigma_0, y_{-n+2} + R_\omega + \varsigma_0)$ and 
\begin{gather*}
s := 2R' + 2R_\omega + 2 \varsigma_0 \geq |y_{-n+2}-y_{-n}|+ 2R_\omega + 2\varsigma_0.
\end{gather*} 
By repetitivity, any interval of length at least $M(s)$ contains a translate $I + t$, $ \omega \cap (I+t) = (\omega \cap I) + t $,
and by pattern equivariance,
\begin{gather*}
\forall\, x,y \in [y_{-n},y_{-n+2}], \quad \ E(x+t,y+t) = E(x,y).
\end{gather*}
Define $R := M(s)$. We claim that $|y_{-k+1} - y_{-k} | \leq R$ for every $1\le k \le n-3 $. 
The proof is by contradiction. Indeed, if this is not the case, based on the foregoing there exists $t\geq0$ such that
\begin{gather*}
\begin{cases}
[y_{-n}-R_\omega-\varsigma_0,y_{-n+2}+R_\omega+\varsigma_0] + t \subseteq (y_{-k},y_{-k+1}),  \\
\forall\, x,y \in [y_{-n}, y_{-n+2}], \quad E(x+t,y+t) = E(x,y).
\end{cases}
\end{gather*}
Aubry crossing lemma (lemma 22 of \cite{GaribaldiPetiteThieullen2017}) shows that
\begin{align*}
E(y_{-k}, y_{-k+1}) &+ E(y_{-n}, y_{-n+1}, y_{-n+2}) = \\
&= E(y_{-k}, y_{-k+1}) + E(y_{-n}+t, y_{-n+1}+t, y_{-n+2}+t)  \\
&> E(y_{-k}, y_{-n+1}+t, y_{-k+1}) + E(y_{-n}+t, y_{-n+2}+t) \\
&= E(y_{-k}, y_{-n+1}+t, y_{-k+1}) + E(y_{-n}, y_{-n+2}).
\end{align*}
Shifting $y_{-n+1}$ to the position $y_{-n+1}+t$, one obtains
\begin{align*}
T^n[0](y_0) & = E(y_{-n}, y_{-n+1}, y_{-n+2}, \ldots , y_{-k}, y_{-k+1}, \ldots , y_0)  \\
& > E(y_{-n}, y_{-n+2}, \ldots , y_{-k}, y_{-n+1}+t, y_{-k+1}, \ldots, y_0).
\end{align*}
We have obtained a configuration of $n+1$ points ending at $ y_0 $ that decreases strictly $T^n[0](y_0)$. 
That contradicts the optimality of $(y_{-n}, \ldots, y_0)$.
\end{proof}

We gather in the following lemma several conclusions that are proved in the lemmas 41 and 42 in \cite{GaribaldiPetiteThieullen2017}. 
We actually simplify the proof and we only use the results of that work exclusively related to the twist condition to obtain the lemma below in a more general framework.

\begin{lemma} \label{Lemma:AprioriEstimateFundamentalConfiguration}
Let $E$ be a weakly twist interaction that is pattern equivariant with respect to a quasi-periodic set $ \omega $. 
Assume $\inf_x E(x,x) > \bar E$. 
Then there exist constants $ \phi>0 $, $ R>0 $ and an integer $ N > 0 $ such that, for every $n\geq1$, for every $(z_0, \ldots, z_n) \in \Gamma_n(E)$,
\begin{enumerate}
\item \label{Item:AprioriEstimateFundamentalConfiguration_2} $|z_n-z_0| \geq n \phi$, 
\item \label{Item:AprioriEstimateFundamentalConfiguration_1} $ \forall\, 0 \le i < n, \  | z_{i+1}-z_i | \leq R$,
\item \label{Item:AprioriEstimateFundamentalConfiguration_3} if $n \geq N$, then $(z_0, \ldots, z_n)$ is strictly monotone.
\end{enumerate}
\end{lemma}

\begin{proof}  \

{\it Item \ref{Item:AprioriEstimateFundamentalConfiguration_2}.} 
Inequality~\eqref{Equation:WeaklyTwistAdditionalProperties} shows that
\[
n \leq A_0 |z_n-z_0| + B_0 \big[ E(z_0,\ldots, z_n)  - n\bar E\big].
\] 
As $(z_0,\ldots,z_n)$ is a fundamental configuration, one has $ E(z_0,\ldots, z_n)  \leq n\bar E $, so that
$ |z_n - z_0| \geq n \phi $, with $\phi = 1/A_0$.

\medskip

{\it Item \ref{Item:AprioriEstimateFundamentalConfiguration_1}.} 
For $L > 0$ as in lemma \ref{Lemma:BoundedJumpBackwardOptimal}, denote $N := \lceil \frac{L}{\phi} \rceil$. 
If $n\geq N$, then $|z_n-z_0| \geq L$ and item \ref{Item:AprioriEstimateFundamentalConfiguration_1} is a consequence of  lemma \ref{Lemma:BoundedJumpBackwardOptimal}. 
If $n \leq N$, let $E^{sup} := \sup_x E(x,x)$ and  $E^{inf} := \inf_{x,y} E(x,y)$.
The superlinearity provides the existence of $ R > 0 $ such that $ |x-y| > R $ implies
\begin{equation*}
E(x,y) > N( E^{sup} - E^{inf} ) + E^{sup}.
\end{equation*}
By contradiction, assume $|z_{-k+1}-z_{-k}|>R$ for some $1 \le k \le n $. 
In particular, $ E(z_{-k}, z_{-k+1}) - E^{sup} > N( E^{sup} - E^{inf} ) $.
Then, using the \textit{a priori} bound $E(z_{-\ell}, z_{-\ell+1}) \geq E^{inf}$ for $\ell \not= k$, as well as 
$\bar E \leq E^{sup}$, we obtain the contradiction
\begin{equation*}
0 \geq E(z_0,\ldots, z_n)  - n\bar E > (N-n+1)( E^{sup} - E^{inf} )\geq0.
\end{equation*}

\medskip
{\it Item \ref{Item:AprioriEstimateFundamentalConfiguration_3}.}  If $n \geq N$, then $|z_n - z_0| \geq L$ and $(z_0, \ldots, z_n)$ is  strictly monotone thanks to lemma  \ref{Lemma:BoundedJumpBackwardOptimal}.
\end{proof}

In the next result, we highlight fundamental properties of the growth of the Ma\~n\'e potential when the quasi-periodic set is linearly repetitive (recall definition~\ref{RepetitiveDeloneSet}). Notice that item \ref{Item:FundamentalConfiguration_2} introduces a dichotomy on the order of fundamental configurations of large size.

\begin{proposition} \label{Proposition:FundamentalConfiguration}
Let $E(x,y) $ be a weakly twist interaction that is pattern equivariant  with respect to a linearly repetitive quasi-periodic set $\omega$. 
Assume that $ \inf_{x}E(x,x) > \bar E $. 
Suppose there exists a sequence $\{K(n)\}_n$ of positive integers diverging to infinite 
for which, associated with each~$n$, there is an increasing fundamental configuration $\underline{z}^n = (z_0^n, \ldots,z_{K(n)}^n)$. 
Then
\begin{enumerate}
\item \label{Item:FundamentalConfiguration_1} the Ma\~n\'e potential has sublinear growth for positively ordered variables
\[
\forall\, \alpha>0, \ \exists\, \beta\geq0, \  \forall\, x \leq y \in \mathbb{R}, \ |S(x,y)| \leq \alpha |y-x| + \beta;
\]
\item \label{complemento_comportamento_potencial} the Ma\~n\'e potential grows linearly for negatively ordered variables
\[
\exists\, \gamma>0, \ \delta\geq0, \ \forall\, x \geq y, \  S(x,y) \geq \gamma |y-x| -\delta.
\]
\item \label{Item:FundamentalConfiguration_2} for $m$ large enough, any fundamental configuration of size $m$ is increasing.
\end{enumerate}
Similarly if  there exists a sequence of decreasing fundamental configurations whose sizes tend to infinity, then the Ma\~n\'e potential has sublinear growth for negatively ordered variables
\[
\forall\, \alpha>0, \ \exists\, \beta\geq0, \  \forall\, x \geq y \in \mathbb{R}, \ |S(x,y)| \leq \alpha |y-x| + \beta,
\]
the Ma\~n\'e potential grows linearly for positively ordered variables
\[
\exists\, \gamma>0, \ \delta\geq0, \ \forall\, x \leq y, \  S(x,y) \geq \gamma |y-x| -\delta,
\]
and all sufficiently long fundamental configuration is decreasing.
\end{proposition}

\begin{proof} 
We prove the case where fundamental configurations are increasing. The other case  will be deduced from the symmetric interaction   $ \hat E(x,y) = E(-x, -y) $.

{\it Item \ref{Item:FundamentalConfiguration_1}.} 
Proposition~\ref{Lemma:WeaklyTwistAdditionalProperties} shows that $S$ admits negative sublinear lower bounds. It is enough to show that $S$ also admits positive sublinear upper bounds. 

Let $C>0$ be the constant in lemma \ref{lemma:aprioriLipschitzLarge} that gives \textit{a priori} growth of $S$, that is,
\[
\forall\, x, y \in\mathbb{R}, \ |S(x,y)| \leq C(|y-x|+1).
\]
The repetitivity assumption says there exist constants $A,B >0$ such that
\[
\forall\, R>0, \ M(R) \leq AR +B,
\]
where $M$ is the repetitivity function introduced in definition~\ref{RepetitiveDeloneSet}. 
As we will avoid working with overlaps, there is no loss of generality in assuming $ A > 1 $. 
Denote
\[
\alpha_* := \inf \big\{ \alpha>0 : \exists\, \beta>0, \ \forall\,  x \leq y, \ S(x,y) \leq \alpha |y-x| +\beta \big\}.
\]
We want to show that $\alpha_*=0$. By contradiction, assume $\alpha_*>0$. 
Let $\alpha \in(\alpha_*, 2\alpha_*)$.
We will reach an absurd by considering a large index $n $ (to be completely defined later) and a corresponding fundamental configuration $(z^n_0, \ldots, z^n_{K(n)})$. 
Initially, applying lemma~\ref{Lemma:AprioriEstimateFundamentalConfiguration}, we require that $ n $ be large enough so that  $(z^n_0, \ldots, z^n_{K(n)})$ is strictly increasing
and, for some $\phi>0$ (see Lemma \ref{Lemma:AprioriEstimateFundamentalConfiguration} for the definition of $\phi$), 
\[
z^n_{K(n)} - z^n_0 \geq  \phi K(n) > 2 \varsigma_0.
\]
Denote 
$$ I_0 := (z^n_0,z^n_{K(n)}) \quad \text{ and } \quad \texttt{P} := (z^n_0 - R_\omega, z^n_{K(n)} + R_\omega) \cap \omega. $$
Note that the pattern $ \texttt{P} $ has diameter $ \ell = z^n_{K(n)} - z^n_0 + 2 R_\omega $ greater than $  2 \varsigma_0 $.
By repetitivity, we may find a sequence $(t_k)_{k\in\mathbb{Z}} \subset \mathbb R $ such that, for all $ k $,  $ \texttt{P} + t_k $ is a pattern of $ \omega $ and
\[
0 < \min (\texttt{P} + t_{k + 1}) - \max (\texttt{P} + t_k) \le M(\ell) - \ell.
\]
Define $(a_k, b_k) = I_k :=I_0+t_k$ and let $I'_k =(b_k,a_{k+1})$ be the interval in between $I_k$ and $I_{k+1}$. 
Then
\begin{align*}
|I'_k| &\leq M(\ell) - \ell + 2(\varsigma_0 + R_\omega) \\ 
&\le (A-1) \ell + B + 2(\varsigma_0 + R_\omega)  = A'|I_k| + B',
\end{align*}
where $ A' = A-1  $ and $ B' = 2 A (\varsigma_0 + R_\omega)  + B $.  
Note thus that, for $ p < q $, 
\[
|a_q - a_p|  = \sum_{k=p}^{q-1} \big( |I_k| + |I'_k| \big) \leq A  \sum_{k=p}^{q-1} |I_k|  + B' (q-p). 
\]
By pattern equivariance,
\begin{align*}
S(a_k,b_k) & \leq E(z^n_0+t_k,\ldots, z^n_{K(n)}+t_k) -K(n) \bar E \\
& =  E(z^n_0, \ldots, z^n_{K(n)}) - K(n) \bar E \leq 0.
\end{align*}
 By sub-additivity of $S$,
\begin{gather*}
S(a_p,a_q) \leq  \sum_{k=p}^{q-1} [ S(a_k,b_k) + S(b_k, a_{k+1}) ] \leq \sum_{k=p}^{q-1} S(b_k,a_{k+1}).
\end{gather*}
By the choice of $\alpha$, there is $ \beta > 0 $ such that
\[
S(b_k,a_{k+1}) \leq \alpha |I'_k| + \beta.
\]
These estimates provide
\begin{align*}
S(a_p,a_q) +\frac{\alpha}{A}|a_q-a_p| 
& \le \sum_{k=p}^{q-1} S(b_k,a_{k+1}) + \alpha \sum_{k=p}^{q-1} |I_k| + \frac{\alpha B'}{A} (q-p) \\
& \le  \alpha |a_q - a_p| + \Big(\frac{\alpha B'}{A} + \beta \Big)(q-p).
\end{align*}
Since $ |a_q-a_p| \geq \sum_{k=p}^{q-1} |I_k| = (q-p) |I_0| $, we obtain
\begin{align*}
S(a_p,a_q) &\leq  \Big( \frac{\alpha A'}{A'+1}  +\Big( \frac{\alpha B'}{A'+1}  + \beta \Big) \frac{1}{|I_0|} \Big)  |a_q-a_p|.
\end{align*}
The distance between $a_k$ and $a_{k+1}$ is at most 
\[
|I_k|+|I_k'| \leq (A'+1)|I_0| +B' := H_0.
\]
If $ y-x > H_0 $ are any given points, we choose $p<q$ such that $x \in \overline{I_p} \cup I'_p$ and $y \in \overline{I_{q-1}} \cup I'_{q-1}$. 
Hence, by the sub-additivity and the \textit{a priori} growth of the Ma\~n\'e potential,
\begin{align*}
S(x,y) & \leq S(x,a_p) + S(a_p,a_q) + S (a_q,y) \\
& \leq \Big( \frac{\alpha A'}{A'+1} +\Big( \frac{\alpha B'}{A'+1} + \beta \Big) \frac{1}{|I_0|} \Big)  |a_q-a_p| +2C(H_0+1) \\
& \le \alpha'  |a_q-a_p|  + \beta',
\end{align*}
with
\begin{gather*}
\alpha' =  \frac{\alpha A'}{A'+1}+\Big( \frac{2\alpha_* B'}{A'+1} + \beta \Big) \frac{1}{|I_0|} \quad \text{ and } \quad \beta' = 2C( H_0+1).
\end{gather*}
Suppose first $ y-x > n H_0 $. Then, using $0 \le x-a_p \le H_0 $ and  $ 0 \le a_q-y \leq H_0 $, we see that $ |a_q-a_p| \leq \big( 1+ \frac{2}{n} \big) |y-x| $ and therefore
\[
S(x,y) \leq \alpha' \Big( 1+ \frac{2}{n} \Big) |y-x| + \beta'.
\]
If however $y-x \leq n H_0$, then
\[
S(x,y) \leq C(n H_0+1) \leq  C(n H_0+1) + \beta' =: \beta''.
\]
We focus on a strictly bigger constant $ \frac{A'}{A'+1} < \frac{2A'}{2A'+1} < 1 $
to choose $\alpha$ sufficiently close to $\alpha_*$ and then $n$ large enough so that
\[
\alpha'' := \Big( \frac{A'}{A'+1} \alpha +\Big( \frac{2\alpha_* B'}{A'+1} + \beta \Big) \frac{1}{\phi K(n)} \Big) \Big( 1+ \frac{2}{n} \Big) < \frac{2A'}{2A'+1}\alpha_*. 
\]
We have obtained two constants $0 < \alpha'' < \alpha_*$  and $\beta'' >0$ such that
\[
 \forall\,  x \leq y, \ S(x,y) \leq \alpha'' |y-x| +\beta''.
\]
The existence of $\alpha''$ contradicts the definition of $\alpha_*$. We have thus proved that $\alpha_*=0$.

\medskip
{\it Item \ref{complemento_comportamento_potencial}.} 
Let $x > y$. 
Item~\ref{Item:WeaklyTwistAdditionalProperties_2} of proposition~\ref{Lemma:WeaklyTwistAdditionalProperties} 
shows that there exist $n\geq1$ and a strictly decreasing sequence $(y_0,\ldots,y_n)$, with $y_0=x$ and $y_n=y$, 
such that
$ S(x,y) = E(y_0, \ldots, y_n) $.
Lemma~\ref{concordancia fundamentais e calibradas} shows that there exists a strictly increasing Ma\~n\'e calibrated configuration $(x_k)_{k\in\mathbb{Z}}$ with bounded jumps that is unbounded from above and bellow. 
Let $x_i \leq y$ be the largest point of this configuration less than or equal to $y$. 
Let $x_j \leq x$ be defined similarly. Then $i \leq j$. If $ i = j $, item~\ref{item 3.Lemma:ManePotential} of lemma~\ref{Lemma:ManePotential} provides 
$ S(x_i, x_j) = E(x_i,x_j) - \bar E $. Otherwise, by calibration $ S(x_i,x_j) =  E(x_i, \ldots, x_j) - (j-i) \bar E $.
Consider now the configuration $(y_0, \ldots,y_n,x_i, \ldots, x_j,y_0)$. Then lemma~23 in \cite{GaribaldiPetiteThieullen2017} guarantees that
\begin{multline*}
S(x,y) + \big( E(y,x_i) - \bar E \big) + S(x_i,x_j) + \big( E(x_j,x) - \bar E \big) \ge \\
\geq (n+j-i+2) \big( \inf_{x \in\mathbb{R}} E(x,x) - \bar E \big).
\end{multline*}
Since $S$ has sublinear growth for positively ordered variables thanks to the previous item, 
for $\alpha>0$ (to be chosen later), there exists $\beta\geq0$ such that
\begin{gather*}
S(x_i,x_j)  \leq \alpha |x_j - x_i| + \beta.
\end{gather*}
As the jumps are bounded, $|x_j - x_i| \leq R(j-i)$, we thus have
\begin{gather*}
S(x,y) +\alpha |x_j - x_i| + \beta  \geq 2\gamma |x_j - x_i| -\delta',
\end{gather*}
with $ 2\gamma := \inf_{x \in\mathbb{R}} ( E(x,x) - \bar E)/R$ and $\delta' := 2\sup_{|y-x| \leq R}\big( E(x,y) - \bar E \big)$. 
We conclude by choosing $\alpha = \gamma $ and $\delta = \delta'+\beta + \gamma R$, so that
\begin{gather*}
S(x,y) \geq \gamma|y-x| - \delta.
\end{gather*}

\medskip

{\it Item \ref{Item:FundamentalConfiguration_2}.} 
Let $m \geq N$ (where $N$ is given in lemma~\ref{Lemma:AprioriEstimateFundamentalConfiguration}) and $(z_0, \ldots, z_m)$ be a decreasing fundamental configuration. 
Let $R>0$ be the constant given in lemma~\ref{Lemma:AprioriEstimateFundamentalConfiguration} and
\[
\alpha := \frac{1}{2R}\big( \inf_{x \in\mathbb{R}} E(x,x) - \bar E \big).
\] 
Then on the one hand, thanks to item \ref{Item:FundamentalConfiguration_1}, there exists $\beta\geq0$ such that
\[
S(z_m,z_0) \leq \alpha |z_0-z_m| + \beta.
\]
On the other hand, thanks to item \ref{Item:WeaklyTwistAdditionalProperties_2} of proposition~\ref{Lemma:WeaklyTwistAdditionalProperties}, one can find an increasing configuration $(x_0, \ldots, x_n)$, with  $x_0=z_m$ and  $x_n = z_0$, such that
\[
S(z_m,z_0) =E(x_0, \ldots,x_n) - n \bar E.
\]
Using item 1 of lemma 23 in \cite{GaribaldiPetiteThieullen2017}, we obtain
\begin{align*}
S(z_m,z_0) &\geq S(z_m,z_0) + E(z_0, \ldots,z_m) - m \bar E \\
&= E(x_0, \ldots,x_n,z_1,\dots,z_m) - (m+n) \bar E \\
&\geq (m+n) \big( \inf_{x \in\mathbb{R}} E(x,x) - \bar E \big) \geq m \big( \inf_{x \in\mathbb{R}} E(x,x) - \bar E \big).
\end{align*}
As item \ref{Item:AprioriEstimateFundamentalConfiguration_1} of lemma~\ref{Lemma:AprioriEstimateFundamentalConfiguration} implies $|z_0-z_m| \leq mR$, 
the choice of $ \alpha $ is contradicted by the inequality
\[
\alpha mR + \beta \geq m \big( \inf_{x \in\mathbb{R}} E(x,x) - \bar E \big)
\]
for $m$ large enough.
\end{proof}

 Note that  the proof makes extensive use  of the fact that the growth rate at infinity of the repetititvity function is linear whereas that  of the Ma\~n\'e potential $S$ is at most linear. We leave open the cases of other growth orders for the repetitivity function, which would, \textit{a priori}, need  stronger conditions on the Ma\~n\'e potential to get a similar result. Recall nevertheless the linear growth rate of the repetitivity function  is the smallest one among aperiodic sets  \cite{Alisteetal}.   

The previous result makes it clear that, in the linearly repetitive context, the ordering of arbitrarily long fundamental configurations plays a key role, thus introducing a preferential ordering to the model.

\begin{definition} \label{Definition:PreferredOrdering}
Suppose that $E(x,y) $ is a weakly twist interaction which fulfills $\inf_{x}E(x,x) > \bar E$, and is pattern equivariant with respect to a linearly repetitive quasi-periodic set $\omega$. 
\begin{enumerate}
\item {\it The preferred ordering} is  the ordering given by sufficiently long fundamental configurations.
\item A monotone configuration is said to be {\it compatible with  the preferred ordering} or for short {\it compatible} if the configuration is ordered  as any fundamental configurations of size sufficiently large. Otherwise, the configuration is said to be {\it anti-compatible with  the preferred ordering} or for short {\it anti-compatible}.
\end{enumerate}
\end{definition}

We now classify the set of weak KAM solutions $u$. There are three approaches: a classification using the type of growth (sublinear \textit{versus} linear), a classification using the ordering of $u$-calibrated subconfigurations, and a classification using bi-infinite $u$-calibrated configurations.  We recall that any calibrated configuration for a weak KAM solution is also Ma\~n\'e calibrated and therefore strictly monotone with a minimal spacing as stated by proposition~\ref{Lemma:WeaklyTwistAdditionalProperties}. 

\begin{lemma}\label{naoaninhados}
Let $E(x,y) $ be a weakly twist interaction that is pattern equivariant  with respect to a linearly repetitive quasi-periodic set $\omega$. 
Suppose that $\inf_{x}E(x,x) > \bar E$. There exists $ L > 0 $ such that,
given any weak KAM solution $ u $ and two points $ x_0 > y_0 + L $, there cannot exist 
simultaneously an increasing $u$-calibrated configuration ending at $x_0$ and a decreasing $u$-calibrated configuration ending
at $ y_0 $.
\end{lemma}

\begin{proof}
Suppose that $ (x_{-m}, \ldots, x_0) $ is increasing, $ (y_{-n}, \ldots, y_0) $ is decreasing, with $ x_0 > y_0 $, and that $ m $ and $ n $ have been chosen so that $ y_0 - R \le x_{-m} \le y_0 $ and $ x_0 \le y_{-n} \le x_0 + R $. 
From item~\ref{complemento_comportamento_potencial} of proposition~\ref{Proposition:FundamentalConfiguration}, we have 
\begin{align*}
u(y_0) - u(y_{-n}) & = S(y_{-n}, y_0) \ge \gamma | y_{-n} - y_0 | - \delta \\
& \ge \gamma | x_0 - x_{-m} | - 2 \gamma R - \delta \ge \gamma (x_0 - y_0) - 2 \gamma R - \delta.
\end{align*}
Recall any weak KAM solution is Lipschitz continuous, with Lipschitz constant bounded by some fixed value $ K $ (lemma~\ref{Lemma:AprioriBoundWeakKAMsolution}). Hence, applying proposition~\ref{Proposition:FundamentalConfiguration} with $ \alpha = \gamma / 2 $ we see that for some $ \beta > 0 $,
\begin{align*}
u(y_0) - u(y_{-n}) & \le 2 KR + u(x_{-m}) - u(x_0) = 2KR - S(x_{-m}, x_0) \\ 
& \le 2KR + | S(x_{-m}, x_0) |  \le 2KR + \frac{\gamma}{2} |x_0 - x_{-m}| + \beta \\
& \le 2KR + \frac{\gamma}{2} (x_0 - y_0) + \frac{\gamma}{2} R  + \beta.
\end{align*}
Therefore,
$$ x_0 - y_0 \le \frac{2}{\gamma} \Big( 2KR +  \frac{\gamma}{2} R + \beta + 2\gamma R + \delta \Big). $$
It is then enough to take $ L := \big \lceil 5R + (4KR + 2\beta + 2\delta)/\gamma \big \rceil + 1 $.
\end{proof}

We introduce vocabulary to quickly refer to possible classes of solutions.

\begin{definition}
Let $u: \mathbb{R} \to \mathbb{R}$ be a function.
We say that  
\begin{enumerate}
\item $u$ is of type I (or $u$ has {\it a sublinear variation})  if 
\begin{gather}
\forall\, \alpha>0, \ \exists\, \beta \geq 0, \ \forall\, x, y \in  \mathbb{R}, \ |u(x)-u(y)| \leq \alpha |x-y| + \beta;\label{Equation:StrictlySublinearVariation}
\end{gather}
\item $u$ is of type II (or $u$ is {\it linearly decreasing with respect to the preferred ordering})  if
\begin{gather*}
\exists\, \gamma, \delta >0, \  \forall\, x,y \in  \mathbb{R}, \ \big[ x \text{ precedes } y \  \Rightarrow \ u(y)-u(x) \leq -\gamma |x-y| + \delta \big];
\end{gather*}
\item $u$ is of type III (or $u$ is of mixed type) it it is of type~I on points that precede $ 0 $ and of type II on points that succeed $ 0 $ according to the preferred ordering.
\end{enumerate}
\end{definition}

\begin{proposition} \label{Proposition:WeakKAMsolutionClassification}
Let $E(x,y) $ be a weakly twist interaction that is pattern equivariant  with respect to a linearly repetitive quasi-periodic set $\omega$. 
Suppose that $\inf_{x}E(x,x) > \bar E$. Let $ u $ be a weak KAM solution.
Then, the following properties hold.
\begin{enumerate}
\item \label{Item:WeakKAMsolutionClassification_1} If $u$ is of type I, any $u$-calibrated configuration is compatible.
\item \label{Item:WeakKAMsolutionClassification_2} If $u$ is of type II, any $u$-calibrated configuration is anti-compatible.
\item \label{Item:WeakKAMsolutionClassification_3} If  $u$ is of type III, there is $T\geq0$ such that 
any $u$-calibrated configuration ending at $x\le -T$ is increasing and ending at $x\ge T$ is decreasing.
\end{enumerate}
\end{proposition}

An equivalent way to describe the role of $ T $ in the last item above is to say that, according to the preferred ordering,
any $u$-calibrated configuration that ends at a point preceding~0 by at least $ T $ is compatible, as well as
any  $u$-calibrated configuration that ends at a point succeeding 0 by at least $ T $ is anti-compatible.

\begin{proof}
Assume the preferred ordering is the increasing order. 

\medskip
{\it Item \ref{Item:WeakKAMsolutionClassification_1}.} 
By contradiction, assume $(x_{-1}, x_0)$ is $u$-calibrated and decreasing. We can extend it to an arbitrarily large decreasing $u$-calibrated configuration $(x_{-n},\ldots, x_{-1}, x_0)$.
From item~\ref{complemento_comportamento_potencial} of proposition~\ref{Proposition:FundamentalConfiguration}, there are $\gamma_1>0$ and $\delta_1\geq0$  fulfilling
\begin{gather*}
u(x_0) - u(x_{-n}) = S(x_{-n}, x_0) \geq  \gamma_1 |x_{-n} - x_0| - \delta_1.
\end{gather*}
However, as $u$ is of type I, it has sublinear variation, so that for $\alpha = \gamma_1/2 $ there is $\beta_1\geq0$ such that
\begin{gather*}
|u(x_0) - u(x_{-n})| \leq \frac{\gamma_1}{2} |x_{-n} - x_0| + \beta_1.
\end{gather*}
As $ |x_{-n} - x_0| \to \infty $ as $ n \to \infty $, we have reached a contradiction. 

\medskip
{\it Item \ref{Item:WeakKAMsolutionClassification_2}.} 
Suppose, by reduction to the absurd, that  $(x_{-1}, x_0)$ is an increasing $u$-calibrated pair, and extend it to an arbitrarily large increasing $u$-calibrated configuration $(x_{-n}, \ldots, x_{-1}, x_0)$. As $u$ is of type II, there are $\gamma_2>0$ and $\delta_2\geq0$ 
\begin{gather*}
-S(x_{-n}, x_0) = u(x_{-n}) - u(x_0) \geq  \gamma_2 |x_0 - x_{-n}| -\delta_2.
\end{gather*}
Since $S$ is sublinear for positively ordered variables, for some $ \beta_2 > 0 $,
\begin{gather*}
|S(x_{-n}, x_0)| \leq \frac{\gamma_2}{2} |x_0 - x_{-n}| + \beta_2.
\end{gather*}
We obtain a contradiction similar to the first item.

\medskip
{\it Item \ref{Item:WeakKAMsolutionClassification_3}}. Let us consider $ T $ of the form $ T := R N $, 
where $ R > 0 $ is the upper bounded for successive jumps of Ma\~n\'e calibrated configurations obtained in proposition~\ref{Lemma:WeaklyTwistAdditionalProperties} and $ N > 1$ will be chosen later. 
Assume that $(x_{-1},x_0)$ is a decreasing  $u$-calibrated pair such that $x_0\le-T$.
Extend it to a decreasing  $u$-calibrated configuration $(x_{-N}, \ldots, x_{-1}, x_0)$. 
In particular, $x_{-N}\leq 0$. As $u$ is of type I on $\mathbb{R}^-$, 
the computations in the first part imply that
\begin{gather*}
Nr \leq |x_{-N} - x_0|  \leq  \frac{2}{\gamma_1} (\beta_1 + \delta_1).
\end{gather*}
Denote then $ N_1 := \lceil \frac{2}{\gamma_1 r}  (\beta_1 + \delta_1) \rceil + 1 $.
Similarly,  if $(x_{-N}, \ldots ,x_0)$ is increasing  $u$-calibrated configuration such that $x_0\ge T$, one obtains $x_{-N} \ge 0$. 
As $u$ is of type II on $\mathbb{R}^+$, the computations of the second part lead us to take into account  $ N_2 := \lceil \frac{2}{\gamma_2 r}  (\beta_2 + \delta_2) \rceil + 1 $. It is enough to choose $ N := \max\{N_1, N_2\} $.
\end{proof}

\begin{proposition} \label{Proposition:WeakKAMsolutionClassificationbis}
Let $E(x,y) $ be a weakly twist interaction that is pattern equivariant with respect to a linearly repetitive quasi-periodic set $\omega$. 
Suppose that $\inf_{x}E(x,x) > \bar E$. Then
\begin{enumerate}
\item  there are examples of weak KAM solutions of each of the three types.
\item every weak KAM solution $ u $ is of one of three types:
\begin{enumerate}
\item \label{Item:WeakKAMsolutionClassification_21}  $u$ is of type I if there exists a bi-infinite $u$-calibrated compatible configuration;

\item \label{Item:WeakKAMsolutionClassification_22} $u$ is of type II if  there exists a bi-infinite $u$-calibrated  anti-compatible configuration;

\item \label{Item:WeakKAMsolutionClassification_23} $u$ is of type III if there is no bi-infinite $u$-calibrated configuration.
\end{enumerate}
\end{enumerate}
\end{proposition}

\begin{proof} 
Assume for instance the preferred ordering is the increasing order. The case of decreasing order is similar.

\medskip
{\it Part 1.} We prove the existence of weak KAM solutions of type I and II.  
Lemma~\ref{concordancia fundamentais e calibradas} shows that one can always ensure the existence of an increasing as well as a decreasing calibrated configuration as required to apply theorem~\ref{theorem:atMostLinearWeakKAMsolution}.
We claim that the resulting weak KAM solutions are of type I and II, respectively.
Recall from lemma~\ref{Lemma:AprioriBoundWeakKAMsolution} that any weak KAM solution is Lipschitz continuous, with
a Lipschitz constant that only depends on the interaction model $ E $. Recall also that the Ma\~n\'e potential is Lipschitz continuous 
as stated in corollary~\ref{potencial Lipschitz}.

The Lipschitz weak KAM solution $ u_I $ obtained from the increasing calibrated configuration inherits the sublinear growth for positively ordered variables of the Ma\~n\'e potential, and it is thus of type I. 
More concretely, for such a solution $ u_I $ we have
\[
\forall\, k < \ell, \quad u_I(x_\ell) - u_I(x_k) = S(x_k,x_\ell)
\]
along the increasing Ma\~n\'e calibrated configuration $ (x_n)_{n \in \mathbb Z} $.
Let $ R > 0 $ denote an upper bound for the successive jumps of a calibrated configuration.
Given $ y - x > 2R $, we may choose $ k $ and $ \ell $ with $ x \in [x_k, x_{k+1}) $ and $ y \in (x_{\ell-1}, x_\ell] $ 
so that
\begin{align*}
 |u_I(x) - u_I(y)| &\leq |S(x_k, x_\ell)| + 2 \text{Lip}(u_I) R \\
 &\leq |S(x,y)| + 2 \big( \text{Lip}(u_I) + \text{Lip}(S) \big) R.
 \end{align*}
 Item~\ref{Item:FundamentalConfiguration_1} of proposition~\ref{Proposition:FundamentalConfiguration} shows that
 the Ma\~n\'e potential has sublinear growth in this situation, we immediately conclude that $ u_I $ has sublinear variation.

By its turn, the solution $ u_{II} $ obtained from the decreasing calibrated configuration is linearly decreasing as a consequence of the behavior of the Ma\~n\'e potential for negatively ordered variables. 
In more precise terms, $ u_{II} $ satisfies
\[
\forall\, k < \ell, \quad u_{II}(y_\ell) - u_{II}(y_k) = S(y_k,y_\ell),
\]
where $ (y_n)_{n \in \mathbb Z} $ is a particular decreasing Ma\~n\'e calibrated configuration.
Similarly as above, for $ y - x > 2 R $ one may find $ \ell > k $ such that
$$  S(y,x) - 2 \big( \text{Lip}(u_{II}) + \text{Lip}(S) \big) R \le S(y_k, y_\ell) - 2 \text{Lip}(u_{II}) R \le u_{II}(x) - u_{II}(y). $$
The fact that $ u_{II} $ is of type II follows thus from item~\ref{complemento_comportamento_potencial}
of proposition~\ref{Proposition:FundamentalConfiguration}.

\medskip
{\it Part 2.} We prove the existence of a weak KAM solution of type III.  
Let $v_I$ and $v_{II}$ be weak KAM solutions of type I and II, respectively. 
We may assume $v_I(0)=v_{II}(0)=0$. 
Define $v_{III} := \min \{v_I, v_{II}\} $. 
Then $v_{III} $ is again a weak KAM solution. 
Let $\gamma_{II}, \delta_{II} > 0 $ be constants used to describe $v_{II} $ as of type II.
Since $ v_I $ is of type I, let $\beta_I > 0 $ be the corresponding constant associated with  $\alpha_I := \gamma_{II} /2$.
For every $x \geq 2(\beta_I+\delta_{II})/\gamma_{II} $, we have the following inequalities
\begin{align*}
& v_I(x) = v_I(x)-v_I(0) \geq  -\alpha_I x -\beta_I, \\  
& v_{II}(x) = v_{II}(x) - v_{II}(0) \leq -\gamma_{II} x +\delta_{II},
\end{align*}
which yield $ v_{II}(x) \leq v_I(x) -\frac{\gamma_{II}}{2}x +\beta_I +\delta_{II} \leq v_{I}(x) $.
Note that whenever $x \leq -2(\beta_I+\delta_{II})/\gamma_{II}$, we get
\begin{gather*}
v_I(x)  \leq  - \alpha_I x + \beta_I, \ \ v_{II}(x) \geq -\gamma_{II} x- \delta_{II}, 
\end{gather*}
so that $ v_I(x)  \leq v_{II}(x) + \frac{\gamma_{II}}{2}x  + \beta_I + \delta_{II} \leq v_{II}(x) $.
To simplify, denote then $T := 2(\beta_I+\delta_{II})/\gamma_{II}$. 
Let $K\geq0$ be the Lipschitz constant of all weak KAM solutions (see lemma~\ref{Lemma:AprioriBoundWeakKAMsolution}). 
If $0 \leq x \leq y$, let $x_T=\max\{T,x\}$, $y_T=\max\{T,y\}$. Then $T \leq x_T \leq y_T$ and
\begin{align*}
v_{III}(y)-v_{III}(x) &\leq v_{III}(y_T) - v_{III}(x_T) + 2KT\\
& = v_{II}(y_T) - v_{II}(x_T) + 2KT \\
&\leq -\gamma_{II} |y_T - x_T| + \delta_{II} + 2KT \\
& \leq - \gamma_{II} |y-x| + \delta_{II} + 2T(K+\gamma_{II}).
\end{align*}
We have proved that $v_{III}$ is of type II on $\mathbb{R}^+$. Given any $\alpha>0$, there is $\beta\geq0$ such that
\begin{gather*}
\forall\, x,y \in\mathbb{R}, \ |v_I(x)-v_I(y)| \leq \alpha |x-y| +\beta.
\end{gather*} 
If $x\leq y \leq 0$, let $x_T=\min\{x,-T\}$, $y_T = \min\{y,-T\}$. Then $x_T \leq y_T \leq -T$ and
\begin{align*}
|v_{III}(x)-v_{III}(y)| &\leq |v_{III}(x_T)-v_{III}(y_T)| + 2KT\\
&= |v_I(x_T)-v_I(y_T)| + 2KT \\
&\leq \alpha|x_T - y_T| + \beta + 2KT \leq \alpha |x-y| + \beta_{III},
\end{align*}
with $\beta_{III} = \beta + 2T(K+\alpha)$. We have proved that $v_{III}$ is of type I on $\mathbb{R}^-$.

\medskip
{\it Part 3.} Conversely, we prove that every weak KAM solution $u$ is one of the three types.  
Let $\mathcal{B}^-_u$ be the set of points $x\in\mathbb{R}$ such that all $u$-calibrated configurations ending at $x$ are increasing. 
Let $\mathcal{B}^+_u$ be the set of points $x$ for which all $u$-calibrated configurations ending at $x$  are decreasing.
From lemma~\ref{naoaninhados}, note that $ \sup (\mathbb R \setminus \mathcal{B}^+_u) \le \inf \mathcal{B}^+_u + L $.
We discuss thus three possibilities.

\medskip
\paragraph{Case $\mathcal{B}^+_u = \mathbb{R}$.} 
 For every $x \in \mathbb{R}$, all $u$-calibrated configurations ending at $x$ are decreasing. 
Given $A>0$, one constructs a decreasing $u$-calibrated configuration ending at $-A$ of size $n$ sufficiently large, $-A = x_0 < x_{-1} < \cdots < x_{-n} $, where $x_{-n}>A$ and $r\leq |x_{-k}- x_{-k+1}| \leq R$ for every $k$. By re-indexing, one has a family of sequences $x^N_{i_N} < \cdots < x^N_{-j_N}$ with $|x^N_0| \leq R$ and $i_N, j_N \to +\infty$. Using a diagonal extraction, one obtains a bi-infinite  decreasing $u$-calibrated configuration $(x_k)_{k\in\mathbb{Z}}$ with bounded jumps and a minimal spacing. For every $x < y - R $, we choose $k \leq \ell$ so that $|x_\ell - x | \leq R$ , $|y -x_k| \leq R$. Let $\gamma, \delta > 0$ be given as in item~\ref{complemento_comportamento_potencial} of proposition~\ref{Proposition:FundamentalConfiguration}. Using the Lipschitz constant $K$ of $u$ and the fact that $(x_k)_{k\in\mathbb{Z}}$ is Ma\~n\'e calibrated, it follows that
\begin{align*}
u(x)-u(y) &\geq u(x_\ell) - u(x_k) - 2KR  = S(x_k,x_\ell) -2KR \\
&\geq  \gamma |x_k - x_\ell| -\delta - 2KR \geq \gamma |y-x| - \delta - 2R(\gamma+K). 
\end{align*}
We have proved that $u$ is of type II.

 \medskip
\paragraph{Case $\mathbb{R} \setminus \mathcal{B}^+_u \not= \emptyset$ is unbounded from above. } 
Note that $\mathbb{R} \setminus \mathcal{B}^+_u $ is the set of points $ x $ for which there exists an increasing
$u$-calibrated  configuration ending at $x$. 
When it is unbounded from above, lemma~\ref{naoaninhados} ensures that $\mathcal{B}^-_u = \mathbb{R}$ and hence all $u$-calibrated configurations are increasing. One can construct as above a strictly increasing bi-infinite $u$-calibrated configuration $(x_k)_{k\in\mathbb{Z}}$ that is also Ma\~n\'e calibrated. 
Let $K$ denote the Lipschitz constant of $u$.
For every $x < y - R $, if $k \leq \ell$ is chosen so that $|x-x_k| \leq R$ and $|y- x_\ell| \leq R$, for every $\alpha>0$ and $\beta\geq0$ given by proposition~\ref{Proposition:FundamentalConfiguration}, one obtains
\begin{align*}
|u(y)-u(x)| &\leq |u(x_\ell) - u(x_k)| + 2KR = |S(x_k,x_{\ell})| + 2KR \\
&\leq \alpha |x_\ell -x_k| + \beta + 2KR \leq  \alpha |y-x| + \beta + 2R(\alpha+K).
\end{align*}
We have proved that $u$ is type I.

\paragraph{Case $b := \sup( \mathbb{R} \setminus \mathcal{B}^+(u)) \in \mathbb R$.}  Let $a := b-L$, where $ L > 0 $ is the constant obtained in lemma~\ref{naoaninhados}. Then  any $u$-calibrated configuration ending at $x_0<a$  is increasing, and any $u$-calibrated configuration ending at $y_0>b$  is decreasing. In particular, there is no bi-infinite $u$-calibrated configuration. 
Fix once for all an increasing $u$-calibrated configuration $(x_{-k})_{k \ge 0}$ ending at a point $ x_0 < \min \{a, 0\} $ and a
decreasing $u$-calibrated configuration $(y_k)_{k \ge 0}$ ending at a point $ y_0 > \max \{b, 0\} $. Both are Ma\~n\'e calibrated.
For $ x < y - R $ with $ y \le x_0 $, making use of the increasing $u$-calibrated configuration as above, given $ \alpha > 0 $ we have that
$$ |u(y)-u(x)| \le  \alpha |y-x| + \beta + 2R(\alpha+K), $$
where $ \beta\geq0$ is guaranteed by proposition~\ref{Proposition:FundamentalConfiguration}. For $ x < y - R $ with $ x \ge y_0 $,
making use of the decreasing $u$-calibrated configuration as in the first case, we see that
$$ u(x)-u(y) \ge \gamma |y-x| - \delta - 2R(\gamma+K), $$
where $\gamma, \delta > 0$ are the constants guaranteed by item~\ref{complemento_comportamento_potencial}
of proposition~\ref{Proposition:FundamentalConfiguration}.
These facts show that $ u $ is of type III.
\end{proof}

We highlight an immediate corollary of the proof of the result above.

\begin{corollary} \label{ordem no primeiro}
Let $E(x,y) $ be a weakly twist interaction that is pattern equivariant with respect to a linearly repetitive quasi-periodic set $\omega$. 
Suppose that $\inf_{x}E(x,x) > \bar E$.  Given a weak KAM solution $u$, with respect to the preferred ordering,
\begin{enumerate}
\item $u$ is of type I if, and only if,  for any  $  y \in \mathbb{R}$, every point belonging to $\argmin \{  u(\cdot) + E(\cdot,y)\} $ precedes $ y $,
\item$u$ is of type II if, and only if,  for any  $  y \in \mathbb{R}$, every point belonging to $\argmin \{  u(\cdot) + E(\cdot,y)\} $ succeeds $ y $,
\item if $u$ is of type III, if and only if, there exists an interval $ I $ of length at most $ L $, where $ L > 0 $ is given as in lemma~\ref{naoaninhados}, for which, for all $ y $ preceding $ I $
every element of $\argmin \{  u(\cdot) + E(\cdot,y)\} $ precedes $ y $, and for all $ y $ succeeding $ I $ every element of $\argmin \{  u(\cdot) + E(\cdot,y)\} $ succeeds $ y $.
\end{enumerate}
\end{corollary}

The next proposition shows that solutions of the same type not only have the same asymptotic behavior characteristics but actually lie at a uniform distance from each other.

\begin{proposition}\label{distancia entre solucoes}
Let $E(x,y) $ be a weakly twist interaction that is pattern equivariant with respect to a linearly repetitive quasi-periodic set $\omega$.  
Assume $\inf_{x}E(x,x) > \bar E$.  
If $ u $ and $ v $ are two weak KAM solutions of the same type, then $\sup_{x \in\mathbb{R}} |u(x)-v(x)| < + \infty$.
\end{proposition}

\begin{proof}
Assume the preferred ordering is the increasing one. There are three cases depending on the types of $u$, $v$.

\medskip
Assume that $u$ and $v$ are of type I. 
Let $(x_k)_{k\in\mathbb{Z}}$ be an increasing bi-infinite $u$-calibrated configuration. 
We may assume $|x_0 | \leq R$, where $R$ is given by proposition~\ref{Lemma:WeaklyTwistAdditionalProperties}. 
We first show by induction that, if $c$ is some constant,  if $u(x_k) \geq v(x_k)+c$ for some $k$, then $u(x_{k+1}) \geq v(x_{k+1})+c$. 
Indeed, using the fact that $v$ is a sub-action, we see that
\begin{align*}
u(x_{k+1}) &= u(x_k) + E(x_k,x_{k+1}) - \bar E \\
&\geq v(x_k) + E(x_k,x_{k+1}) -\bar E + c \geq v(x_{k+1})+c.
\end{align*}
Starting with $c=u(x_0)-v(x_0)$, one thus obtains
\begin{gather*}
\forall\, k \geq 0, \ u(x_k) -v(x_k) \geq u(x_0)-v(x_0).
\end{gather*}
We then extend the above inequality to every  $x > 0$. We choose $k\geq0$   such that $|x - x_k| \leq R$. Using the Lipschitz constant $K$ of $u$ and $v$, one has
\begin{gather*}
\forall\, x \geq 0, \ u(x) - v(x) \geq u(x_k) - v(x_k) -2KR \geq u(x_0) - v(x_0) -2KR. 
\end{gather*}
By permuting the role of $u$ and $v$, one obtains
\begin{gather*}
\sup_{x\geq 0}|u(x)-v(x)| \leq \sup_{y\in[-R,R]} |u(y)-v(y)| + 2KR.
\end{gather*}
Similarly if $u(x_k) \leq v(x_k)+c$ for some $k \leq 0$, by induction one gets $u(x_{k-1}) \leq v(x_{k-1})+c$. 
One concludes as before that
\begin{gather*}
\forall\, x \leq 0, \ u(x)-v(x) \leq u(x_0) - v(x_0)+2KR,  \quad \text{ and} \\
\sup_{x\leq 0}|u(x)-v(x)| \leq \sup_{y\in[-R,R]} |u(y)-v(y)| + 2KR.
\end{gather*}

If $u$ and $v$ are of type II, the proof is completely analogous using decreasing bi-infinite $u$-calibrated and $v$-calibrated configurations.

Assume that $u$ and $v$ are of type III. Let $T\geq0$ be the maximum of the respective constants for $ u $ and $ v $ as given in proposition~\ref{Proposition:WeakKAMsolutionClassification}. Through decreasing $u$-calibrated and $v$-calibrated configurations ending at $T$, in a similar way as above, we get $ \sup_{x \ge T} |u(x)-v(x)| \leq |u(T)-v(T)| + 2KR$. 
Then using increasing $u$-calibrated and $v$-calibrated configurations ending at $-T$, we obtain 
$ \sup_{x \le -T} |u(x)-v(x)| \leq |u(-T)-v(-T)| + 2KR$.
\end{proof}

\section{Periodic interaction models}\label{periodic and linear}

Theorem~\ref{theo:mainThm} illustrates that, when the  non-degeneracy condition is satisfied ($ \inf_x E(x,x) > \bar E $),  the asymptotic behaviour of a weak KAM solution is related to the one of calibrated configurations. A noticeable point in the proof is the existence of unbounded Ma\~n\'e calibrated configurations, at least in one direction. 
We show the non-degeneracy hypothesis  is necessary  for these results by providing an example of a degenerate interaction for which any Ma\~n\'e  calibrated configuration is bounded but admitting a weak KAM solution with linear growth. In  more precise words, we prove in this section theorem~\ref{theo:Periodic degenerate}.  Unexpectedly, the example falls in the classical framework of  a periodic  interaction mechanical model of   Frenkel-Kontorova type. 

\begin{hypotheses}\label{label:periodicMainAssumptions}
Let $E:\mathbb{R} \times \mathbb{R} \to \mathbb{E}$ be an interaction model satisfying all the assumptions (H1-4) of hypotheses~\ref{assumption:interactionModel}. We will say that $E$ is a {\it periodic interaction model}  whenever

\medskip

\noindent ($\text{H}_{\text{per}}$) $E$ is periodic:
$ \, \forall \,x,y \in\mathbb{R}, \quad E(x,y) = E(x+1,y+1) $.
\end{hypotheses}

Periodicity suggests to  introduce another interaction $E^{per} : \mathbb{R} \times \mathbb{R} \to \mathbb{R}$ defined as
\[
E^{\text{per}}(x,y) := \min_{k\in\mathbb{Z}}E(x,y+k). 
\]
Note that $E^{per} : \mathbb{T} \times  \mathbb{T} \to \mathbb{R}$ may be seen as a function on the torus, that is, 
$ E^{per} $ is doubly periodic.
But we will see it is important to keep the real line orientation and consider configurations on $\mathbb{R}$ and not on $\mathbb{T}$. 
It is easy to verify  that the two ground actions coincide,
\[
\bar E = \bar E^{\text{per}}.
\]
As in definition~\ref{definition:basicWeakKAMtools} and in~\eqref{definition:LaxOleinikOperator}, we may similarly introduce periodic 
versions of the Ma\~n\'e potential and Lax-Oleinik operator (now acting on continuous periodic functions):
\begin{gather*}
\forall \,x,y \in \mathbb{R}, \quad S^{\text{per}}(x,y) =  \inf_{n\geq1} \ \inf_{x = x _0,\ldots,x_n=y} \ \sum_{k=0}^{n-1} \big( E^{per}(x_k,x_{k+1}) - \bar E \big), \\
\forall \,u \in C^0(\mathbb{T}), \ \forall \,y \in \mathbb{R}, \quad T^{per}[u](y) := \inf_{x\in\mathbb{R}} \big\{ u(x) + E^{per}(x,y) \big\}.
\end{gather*}

It is easy to check the following properties.

\begin{lemma}
Let $E$ be an interaction model satisfying the hypotheses~\eqref{label:periodicMainAssumptions}.
\begin{enumerate}
\item The periodic Ma\~n\'e potential for the interaction $E^{\text{per}}$ is related to the Ma\~n\'e potential by
\begin{align*}
\forall \,x,y \in \mathbb{R}, \quad S^{\text{per}}(x,y) = \inf_{k \in\mathbb{Z}} S(x,y+k).
\end{align*}
\item If $u \in C^0 (\mathbb{R})$ is 1-periodic, then
\[
T^{\text{per}}[u] = u + \bar E \quad \Leftrightarrow \quad T[u] = u+\bar E.
\]
\item If $u$ is a 1-periodic weak KAM solution (for either $E$ or $E^{\text{per}}$),   the optimal points in  the two definitions of $T$ and $T^{\text{per}}$ coincide modulo $\mathbb{Z}$: for every $y\in\mathbb{R}$, we have
\begin{gather*}
 \argmin_{x\in\mathbb{R}} \{ u(x) + E^{\text{per}}(x,y)-\bar E \} = \argmin_{x\in\mathbb{R}} \{ u(x) + E(x,y)-\bar E \} + \mathbb{Z}. 
\end{gather*}
\end{enumerate}
\end{lemma}

By taking into account the periodic potential
$$ V(x) := 1 - \cos (2\pi x), $$
we will focus on the following version of the standard Frenkel-Kontorova model
\begin{equation}
E_\lambda (x,y) := \frac{1}{2} |y - x - \lambda|^2 +KV(x), \qquad \lambda \in \mathbb{R}, \ K \in \mathbb{R}_+. \label{definicao exemplo modelo periodico}
\end{equation}
 Note that $E_\lambda$ satisfies assumptions (H1-4) of hypotheses~\ref{assumption:interactionModel}, but it may be degenerate.
 Let $ \bar{E_\lambda} $ denote the corresponding ground action.

We now fix  $ K > 0 $ and discuss the properties on the periodic interaction model~\eqref{definicao exemplo modelo periodico} with respect to $ \lambda $.
Denote by $ S_\lambda(x, y) $ the respective Ma\~n\'e potential.  
Since the interaction $E_0$ may be written as a sum of even functions, the Ma\~n\'e potential and the ground action preserve some symmetries:
\begin{eqnarray*}
 S_0(0,y) \!\! &= & \!\!   S_0(0,-y),   \\
 S_0(x,y) + KV(y) \!\!  &=& \!\!   S_0(y,x) + KV(x),  \\
E_{\lambda+1}(x,y) = E_{\lambda}(x,y), &   &   E_{-\lambda}(x,y)= E_{\lambda}(-x,-y), \\
E_\lambda(x,y) -   \frac{1}{2} \lambda^2  &=& E_0(x, y) - \lambda(y - x), \\
\bar{E}_{\lambda +1} = \bar{E}_{\lambda},&    & \bar{E}_{-\lambda}= \bar{E}_{\lambda}. 
\end{eqnarray*}

We resume in the next proposition the main properties of the  Ma\~n\'e potential  and of the Ma\~n\'e calibrated configurations that will give  theorem \ref{theo:Periodic degenerate}. 
More precisely, for suitable parameters $K$ and $\lambda$, this model is degenerate (item \ref{item:periodicManePotential_2}), the  Ma\~n\'e potential defines a weak KAM solution (item \ref{item:periodicManePotential_5}) of linear growth (items \ref{item:periodicManePotential_1} and \ref{item:periodicManePotential_1bis}) but all Ma\~n\'e  calibrated configurations are bounded (item \ref{item:periodicManePotential_7}).

\begin{proposition} \label{proposition:periodicManePotential}
Let $K >0$  and  $  c_0(K) := \min \{\sqrt{1 + 2K} - 1, \frac{1}{40}\}$. Then for every  $|\lambda| <c_0(K)$, we have
\begin{enumerate}
\item \label{item:periodicManePotential_1} $\forall \,\, y \in \mathbb R, \ S_0(0, y)\ge c_0(K) (| y |-\frac{1}{2})$, 
\item \label{item:periodicManePotential_1bis}  $\exists \, c_1(K) > c_0(K), \ \forall \,\, x,y \in \mathbb R, \ S_0(x,y) \leq c_1(K)(|y-x|+1)$,
\item \label{item:periodicManePotential_2} $ \bar E_\lambda = \inf_{x \in \mathbb R} E_\lambda (x, x) = \frac{1}{2}\lambda^2 $, $\bar E_\lambda = E_\lambda(x,x) \ \Leftrightarrow \ x \in \mathbb{Z}$
\item \label{item:periodicManePotential_3}  $S_\lambda(x,y) = S_0(x,y) - \lambda(y-x)$, $S_\lambda(0,y) \geq (c_0(K)-|\lambda|)|y| - \frac{1}{2}c_0(K)$,
\item \label{item:periodicManePotential_4} $\forall \,i < j < k, \ i,j,k \in \mathbb{Z}, \quad S_\lambda(i, k) = S_\lambda(i, j) + S_\lambda(j, k)$,
\item \label{item:periodicManePotential_5} $S_\lambda(0,x)$ is a  uniformly Lipschitz weak KAM solution,
\item \label{item:periodicManePotential_7} if $(x_k)_{k\in\mathbb{Z}}$ is a Ma\~n\'e calibrated configuration for the interaction $E_\lambda$, then there is an integer $n$ such that $x_k \in [n,n+1],  \ \forall \,k\in\mathbb{Z}$.
\end{enumerate}
\end{proposition}

Note that this model is non-degenerate for some parameters. Actually, it will be shown in  proposition~\ref{Exemplo hipotese} of appendix~\ref{Nao degenerescencia} (for $\rho =1$)  that this periodic Frenkel-Kontorova model  is non-degenerate for large $|\lambda|$ compared to $K$, namely:
$$\forall \, 0 \le K < \frac{\lambda^2}{8}, \qquad \bar{E}_{\lambda} < \inf_{x \in {\mathbb R}}  E_{\lambda}(x,x).
$$

\begin{proof}[Proof of proposition \ref{proposition:periodicManePotential}.]

\quad
\medskip

\noindent {\it Item  \ref{item:periodicManePotential_1}}. By the symmetries of $S_0$, we may suppose  that $ y > 0 $. 
Note that $\bar E_0=0$. Thanks to lemma~\ref{Lemma:ManePotential}, for the computation of $S_0(0,y)$, 
it is enough to consider monotone configurations  $ 0 = x_0 < x_1 < \ldots < x_n = y$. Denote 
\[
 \{i_\alpha < i_{\alpha + 1} < \ldots < i_\beta\} := \{ 0\le i  <n : [x_i, x_{i+1}) \cap (\mathbb Z + {\textstyle\frac{1}{2}}) \neq \emptyset\}.
\] 
Clearly, $ 0 \le x_{i_\alpha} < \frac{1}{2} $, $ x_{i_{k+1}} - x_{i_k+1} < 1 $, and 
$ y - x_{i_\beta} < 1 $. By positivity of $E_0$, we have
\begin{equation*}
\sum_{i=0}^{n-1} E_0(x_i, x_{i+1}) \ge \sum_{k=\alpha}^{\beta} E_0(x_{i_k}, x_{i_k+1}). 
\end{equation*}

We claim that there exists a constant $ c \in (0, \frac{1}{4}) $ such that, for any such a subconfiguration, 
$ E_0(x_{i_k}, x_{i_k+1}) > c \, (x_{i_k+1} - x_{i_k} + 1) $, $ \forall \,k = \alpha, \ldots, \beta $. Using  $0 \leq x_{i_{k+1}}-x_{i_k+1} \leq 1$,  the claim will imply
\begin{equation}\label{estimativa central}
\sum_{k=\alpha}^{\beta} E_0(x_{i_k}, x_{i_k+1}) \ge c \, (y- x_{i_\alpha}).
\end{equation}

Let us prove the claim. 
Denote by $ q_k + \frac{1}{2} $ the smallest element of $\mathbb Z + \frac{1}{2}$ belonging to $ [x_{i_k}, x_{i_k+1}) $.  
If we write $ u_k = q_k + \frac{1}{2} - x_{i_k} $ and $ v_k = x_{i_k+1} - (q_k + \frac{1}{2}) $, to show our claim it is enough to assure the existence of a constant $  c \in (0, \frac{1}{4} )  $ such that, for every $u, v \ge 0$,
\begin{equation*}
F(u,v) :=  \frac{1}{2} (u + v)^2 + K (1 + \cos 2 \pi u) - c \, (u + v +1) \geq0. 
\end{equation*}
The minimum of $ F(u, v)$ over $ u, v \ge 0 $ takes place on the boundary of the domain, since there is no critical points in the interior of the domain.
Hence, we have to consider the following three cases.

\smallskip

\noindent {\it Case 1.} On the border $ u = 0 $, note that    
\[
F(0,v) = \frac{1}{2} v^2 - c v - c + 2 K  \ge   2K -\frac{1}{2}(c+1)^2 + \frac{1}{2} \geq 0,
\] 
whenever $c\in [0, \sqrt{1+4K}-1] $.

\smallskip

\noindent {\it Case 2a.} On the border $ v = 0 $ and $ u \ge \frac{1}{4} $, we have
\[
F(u, 0)  \ge \frac{1}{2} u^2 - c \, (u + 1) \geq \frac{1}{32} - c\frac{5}{4} \geq0,
\] 
whenever $c \in [0, \frac{1}{40}]$. 

\smallskip

\noindent {\it Case 2b.} Finally, on the border $ v = 0 $ and $ 0 \le u \le \frac{1}{4} $, observe that 
\[
F(u, 0) \ge \frac{1}{2} u^2 + K  - c \, (u + 1) \geq K - \frac{1}{2}(c+1)^2 +\frac{1}{2} \geq0,
\] 
whenever $c \in  [0, \sqrt{1 + 2K} - 1] $.

\smallskip

To summarize, we have shown that the inequality~\eqref{estimativa central} holds for the constant $ c = c_0(K) = \min\{\sqrt{1 + 2K} - 1, \frac{1}{40}\} $. With the obvious estimate
\begin{gather*}
\sum_{i=0}^{i_\alpha-1} E_0(x_i, x_{i+1})  \ge c_0(K)(x_{i_\alpha}-{\textstyle\frac{1}{2}}),
\end{gather*}
we conclude that $ S_0(0, y) \ge c_0(K) (y-\frac{1}{2})$ for all $ y > 0 $. 

\medskip

\noindent {\it Item \ref{item:periodicManePotential_1bis}.} The existence of $c_1(K)$ follows from lemma \ref{lemma:aprioriLipschitzLarge}.

\medskip

\noindent {\it Item \ref{item:periodicManePotential_2}}. 
On the one hand, as $ \bar E_\lambda = \lim_{n\to+\infty} \inf_{x_0,\ldots,x_n} \frac{1}{n} E_\lambda(x_0,\ldots,x_n) $, from the configuration $x_k=0 $  for all $ k $, one obtains $ \bar E_\lambda \le \frac{1}{2}\lambda^2 $.

On the other hand, let us introduce the function 
$$ u_\lambda (x) = \inf_{n\geq1} \ \inf_{x_0 = 0,\ldots,x_n=x} \Big[ E_\lambda(x_0,\ldots,x_n) - n \frac{\lambda^2}{2} \Big]. $$
Since $ E_\lambda(x,y) -   \frac{ \lambda^2}{2}  = E_0(x, y) - \lambda(y - x) $, clearly $ u_\lambda(x) = S_0(0, x) - \lambda x $, and in particular, $ u $ is a well-defined function.
Note that
\begin{gather}
u_\lambda(x) \ge (c_0(K) - |\lambda|) | x | - {\textstyle\frac{1}{2}}c_0(K) \geq -{\textstyle\frac{1}{2}}c_0(K), \quad \forall \,\, x \in \mathbb R, \notag \\
u_\lambda(x) + E_\lambda(x,y) \ge u_\lambda(y) + {\textstyle\frac{1}{2}} \lambda^2, \quad \forall \,\, x, y \in \mathbb R.  \label{u subacao}
\end{gather}
Define then $ u^{\text{per}}_\lambda(x) := \inf_{p \in \mathbb Z} u_\lambda(x + p) $. Using the constant $c_1(K)$ of item~\ref{item:periodicManePotential_1bis}, it follows that
\[
 -{\textstyle\frac{1}{2}} c_0(K) \le  u^{\text{per}}_\lambda(x) \le \inf_{p \in \mathbb Z} \big[(c_1(K) + |\lambda|)(|x + p| + 1) \big] \le 2(c_1(K) + |\lambda|),
\] 
namely, $ u^{\text{per}}_\lambda $ is bounded.  
From~\eqref{u subacao},  we obtain for every $y\in\mathbb{R}$,
\begin{align*}
u^{\text{per}}_\lambda(y) + {\textstyle\frac{1}{2}} \lambda^2 &\leq \inf_{x\in\mathbb{R}} \ \inf_{p,q \in \mathbb Z}  \big[u_\lambda(x+p) + E_\lambda(x, y+q) \big]   \\
&= \inf_{x\in\mathbb{R}} \big[ u^{\text{per}}_\lambda(x) + E^{\text{per}}_\lambda(x, y) \big].
\end{align*}
Therefore, we clearly get
$$ \bar E_\lambda = \lim_{n\to+\infty} \inf_{x_0,\ldots,x_n \in \mathbb R} \frac{1}{n} \big[E_\lambda(x_0,\ldots,x_n) + u^{\text{per}}_\lambda(x_0) - u^{\text{per}}_\lambda(x_n)\big] \ge {\textstyle\frac{1}{2}}\lambda^2. $$ 

\medskip

\noindent {\it Item \ref{item:periodicManePotential_3}}. It follows from $\bar E_\lambda = \frac{1}{2}\lambda^2$.

\medskip

\noindent {\it Item \ref{item:periodicManePotential_4}}. 
From  item \ref{item:periodicManePotential_3}, it is enough to prove the relation for  $ \lambda = 0 $. In addition, since $ S_0(m, n) = S_0(0, n - m) $ for any integers $ m < n $, one just needs to argue that $ S_0(0, k) = k S_0(0,1) $ for all positive integers $ k $.
Such a fact follows from 
$$ S_0(0, k + 1) = S_0(0, k) + S_0(0,1), \qquad \forall \,\, k \ge 1. $$
To see that the equality above holds, note first that clearly $ S_0(0, k+1) \le S_0(0, k) + S_0(k, k+1) = S_0(0,k)+ S_0(0,1) $. 
Recall that we may consider only monotone configurations $ (x_0 = 0, \ldots, x_n = k + 1) $ in the expression defining $ S_0(0, k+1) $ (lemma~\ref{Lemma:ManePotential}). 
Observe that for $ u, v \ge 0 $ and $ j \in \mathbb Z $,
$$ E_0(j - u, j + v) - [E_0(j - u, j) + E_0(j, j+v)] = \frac{1}{2}(u+v)^2 - \frac{1}{2} u^2 - \frac{1}{2} v^2 \ge 0. $$
For $m = \max\{0\le i \le n: x_i<k \}$, this inequality implies the configuration obtained by concatenating the configurations   
$ (x_0 = 0, \ldots, x_m, k) $ and $(k, x_{m+1}, \ldots, x_n=k+1)$  does not increase the total energy. This ensures the opposite   inequality $ S_0(0, k + 1) \ge S_0(0, k) + S_0(0,1) $.

\medskip

\noindent {\it Item \eqref{item:periodicManePotential_5}.} 
Thanks to item~\eqref{item:periodicManePotential_3}, it is enough to argue that $ S_0(0, \cdot) $ is a weak KAM solution. Note first that $ x \mapsto S_0(0,x) $ is a continuous function. Obviously, for all $ x $ and $ y $
\[
S_0(0, y) - S_0(0,x) \le S_0(x, y) \le E_0(x,y) - \bar E_0 = E_0(x,y).
\]
On the other hand, for each positive integer $ k $, there exists a subconfiguration $ (x_1^k, \ldots, x_{n_k}^k) $ such that
\[
S_0(0, y) + \frac{1}{k} \ge E_0(0, x_1^k, \ldots, x_{n_k}^k, y) \ge 0.
\]
In particular, as
\[
0 \le E_0(x_{n_k}^k, y) \le S_0(0, y) + 1 \qquad \forall \, k \ge 1,
\]
the superlinearity implies the sequence $ \{x_{n_k}^k\} $ is at a bounded distance from $ y $, and therefore admits an accumulation point. Passing to a subsequence if necessary, we then assume that $ x_{n_k}^k \to \bar x $. We have for all~$ k $
\[
S_0(0,y) + \frac{1}{k} \ge S_0(0, x_{n_k}^k) + E_0(x_{n_k}^k, y),
\]
which thanks to the continuity yields 
\[
S_0(0,y) \ge S_0(0, \bar x) + E_0(\bar x, y).
\]

We next show  that the minimum in 
\[
 S_\lambda(0,y) = \min_{x\in\mathbb{R}} \big[ S_\lambda(0,x) + E_\lambda(x,y) - \bar E_\lambda \big],
\]
is attained at some $x$ which satisfies $|y-x| \leq R$ for some $R>0$ independent of $y$. 
The infimum is realized at some $x$ thanks to the superlinearity of $E_\lambda$.
Using item \ref{item:periodicManePotential_1bis}, one obtains
\[
S_\lambda(x,y) \leq (c_1(K)+|\lambda|)(|y-x|+1), \quad \forall \,x,y \in \mathbb{R}.
\]
On the one hand,
\[
S_\lambda(0,y) \leq S_\lambda(0,x) + S_\lambda(x,y) \leq S_\lambda(0,x) +  (c_1(K)+|\lambda|)(|y-x|+1).
\]
On the other hand, thanks to the superlinearity of $E_\lambda$, there exists a constant $c_2(K,\lambda)$ such that
\[
E_\lambda(x,y) -\bar E_\lambda \geq (c_1(K)+|\lambda| +1)|y-x|-c_2(K,\lambda), \quad \forall \,x,y \in \mathbb{R}.
\]
Using the two previous estimates and $S_\lambda(0,y) = S_\lambda(0,x) + E_\lambda(x,y)-\bar E_\lambda$, one gets
\[
|y-x| \leq c_1(K)+|\lambda|+c_2(K,\lambda) := R.
\]
The Lipschitz constant of $y \mapsto S_\lambda(0,y)$ depends on the Lipschitz constant of $y \mapsto E_\lambda(x,y)$ uniformly on $|y-x | \leq R$.

\medskip

\noindent {\it Item \ref{item:periodicManePotential_7}.} Observe that by  item \ref{item:periodicManePotential_3}, a Ma\~n\'e calibrated configuration for $E_\lambda$ is also Ma\~n\'e calibrated for $E_0$. So without lost of generality, we can assume $\lambda=0$. 

{\it Part 1.} We show that $(x_k,x_{k+1})$ cannot contain an integer. By contradiction, if $n \in (x_k,x_{k+1})$, the twist property (see lemma 22 of \cite{GaribaldiPetiteThieullen2017}) implies
\begin{align*}
S_0(x_k,x_{k+1}) &= S_0(x_k,x_{k+1}) + E_0(n,n), \\
&= E_0(x_k,x_{k+1}) + E_0(n,n), \\
&> E_0(x_k,n) + E_0(n,x_{k+1}) \geq S_0(x_k,x_{k+1}),
\end{align*}
and we reach an absurd.

{\it Part 2.} We show there cannot exist an integer $p\in {\mathbb Z}$ such that for some index $k\in \mathbb Z $  and some integer $l\ge 0$,
\[
x_{k-1} < x_k = x_{k+1} = \cdots = x_{k+l} < x_{k+l+1}, \quad x_k =p \in \mathbb{Z}.
\] 
The other case $x_{k-1} > x_k > x_{k+l+1}$ and $x_k \in \mathbb{Z}$ is done similarly. By contradiction, on the one hand, the function 
\[
x \mapsto E_0(x_{k-1},x, \cdots,x,x_{k+l+1}) \qquad \text{(}x \text{ repeated $l+1$ times)}
\] 
reaches its minimum at $x=x_k$ and for small $\epsilon>0$ one has
\[
E_0(x_{k-1},x_k+\epsilon, \cdots,x_{k+l}+\epsilon,x_{k+l+1}) = S_0(x_{k-1},x_{k+l+1}) + O(\epsilon^2).
\]
On the other hand, the  twist property implies there exists $\alpha>0$ (independent of $\epsilon$) such that
\begin{align*}
E_0(x_{k-1},x_k+\epsilon) &= E_0(x_{k-1},x_k+\epsilon) + E_0(x_k,x_k) \\
&\geq E_0(x_{k-1},x_k) + E_0(x_k,x_k+\epsilon) + \alpha \epsilon.
\end{align*}
Thus 
\begin{align*}
S_0(x_{k-1},x_{k+l+1}) &+ O(\epsilon^2) \geq \\
&\geq E_0(x_{k-1},x_k,x_k+\epsilon, \cdots,x_{k+l}+\epsilon,x_{k+l+1}) +\alpha\epsilon \\
&\geq S_0(x_{k-1},x_{k+l+1}) + \alpha \epsilon,
\end{align*}
and we obtain a contradiction.

The two previous parts show that $(x_k)_{k\in\mathbb{Z}}$ cannot overlap an interval of the form $[n,n+1]$, $n\in\mathbb{Z}$.
\end{proof}

As suggested by one of the referees, the way order emerges and the relationship with the non-degeneracy condition may be illustrated through the study of one-parameter periodic families. 
Proposition~\ref{Proposition:ComplementPreferredOrdering} summarizes this study and we conclude this section with its proof.
Recall that $ E_0 $ is a given periodic interaction model and $ E_\lambda(x, y) = E_0(x, y) - \lambda(y-x) $ for $ \lambda \in \mathbb R $.

\begin{proof}[Proof of proposition~\ref{Proposition:ComplementPreferredOrdering}.]

\

\medskip
\noindent {\it Item \ref{Item:ComplementPreferredOrdering_1}.} 
Clearly, $ \inf_{x} E_\lambda(x,x) = \inf_{x} E_0(x,x) $.
Assume
\begin{gather*}
\lambda >  \inf_{x \in\mathbb{R}} E_0(x,x+1) - \inf_{x \in \mathbb{R}} E_0(x,x).
\end{gather*}
Then by taking configurations of the form $(x,x+1, \ldots, x+n)$   where $x \in \argmin_{x} E_0(x,x+1)$, one obtains
\begin{gather*}
\inf_{x_0, \ldots,x_n} E_\lambda(x_0, \ldots, x_n) \leq \sum_{k=1}^n E_\lambda(x + k-1,x+k) = n \big( E_0(x,x+1) - \lambda), \\
\bar E_\lambda \leq  \inf_{x \in\mathbb{R}} E_0(x,x+1) - \lambda  < \inf_{x \in \mathbb{R}} E_0(x,x) = \inf_{x \in \mathbb{R}} E_\lambda(x,x).
\end{gather*}
We have proved that $E_\lambda$ is non-degenerate. Suppose now  
\begin{gather*}
\lambda < \inf_{x \in \mathbb{R}} E_0(x,x) - \inf_{x \in \mathbb{R}} E_0(x+1,x).
\end{gather*}
Then by taking configurations of the form $(x+n, \ldots,x+1,x)$ where $x \in \argmin_{x} E_0(x+1,x)$, one has
\begin{gather*}
\inf_{x_0, \ldots,x_n} E_\lambda(x_0, \ldots, x_n) \leq \sum_{k=1}^n E_\lambda(x+k,x+k-1) = n \big(  E_0(x+1,x)  + \lambda \big), \\
\bar E_\lambda \leq \inf_{x \in \mathbb{R}} E_0(x+1,x)  + \lambda <  \inf_{x \in \mathbb{R}} E_0(x,x) = \inf_{x \in \mathbb{R}} E_\lambda(x,x).
\end{gather*}
We have again showed that  $E_\lambda$ is non-degenerate.

\medskip
\noindent {\it Item \ref{Item:ComplementPreferredOrdering_2}.} Let
\begin{gather*}
\lambda > \inf_{x \in \mathbb{R}}E_0(x,x+1) - \inf_{x,y \in \mathbb{R}} E_0(x,y).
\end{gather*}
Suppose by contradiction there are fundamental configurations $(x_0^n, \ldots, x_n^n)$ of arbitrarily large size $ n $ for $E_\lambda$ that are decreasing. Then, one has
\begin{gather*}
E_\lambda(x_0^n, \ldots, x_0^n) \geq n \inf_{x,y \in \mathbb{R}}E_0(x,y) + \lambda (x_0^n - x_n^n) \geq n  \inf_{x,y\in \mathbb{R}}E_0(x,y), \\
\bar E_\lambda \geq \inf_{x,y \in \mathbb{R}}E_0(x,y).
\end{gather*}
However, the first part shows that
\begin{gather*}
\bar E_\lambda \leq \inf_{x \in\mathbb{R}} E_0(x,x+1) - \lambda.
\end{gather*}
We thus obtain a contradiction. Therefore, a preferred ordering exists and is the increasing one.

\medskip
\noindent {\it Item \ref{Item:ComplementPreferredOrdering_3}.} We change the ordering by taking the new interaction model
\begin{gather*}
\tilde E_\lambda(x,y) := E_\lambda(-x,-y) = \tilde E_0(x,y) -(-\lambda)(y-x).
\end{gather*}

\medskip

\noindent {\it Item \ref{Item:ComplementPreferredOrdering_4}.} 
Note that by its very definition the function
\[
\lambda \mapsto \bar E_\lambda = \lim_{n\to +\infty}\frac{1}{n}\inf_{x_0,\ldots, x_n}[E_0(x_0,\ldots,x_n) - \lambda(x_n-x_0)] 
\] 
is concave. 
Therefore, the parameters $ \lambda $ for which $ E_\lambda $ is degenerated correspond to the convex and closed set $D := \{\lambda \in \mathbb R :  \bar E_\lambda = \inf_x E_0(x,x) \}$, which is an interval eventually empty.  
According to proposition \ref{Proposition:FundamentalConfiguration}, a parameter $\lambda$ whose potential $E_{\lambda}$ is non-degenerated belongs  to  either $\Lambda_{+}$ or $\Lambda_{-}$. Hence the real line  is partitioned into $\mathbb R= \Lambda_{-} \sqcup D \sqcup \Lambda_{+}$. Moreover, the previous items show that both sets $\Lambda_+ \cap \mathbb R_{+}$ and $\Lambda_- \cap \mathbb R_{-}$ are not empty.

We claim that:
\begin{itemize}
\item[(a)] for all $ \lambda \in \Lambda_{+} $, there is $ \epsilon>0 $ such that $ (\lambda-\epsilon, +\infty) \subset \Lambda_{+} $;

\item[(b)] for all $ \lambda \in \Lambda_{-} $, there is $ \epsilon>0 $ such that $ (-\infty, \lambda+\epsilon) \subset \Lambda_{-} $;
\end{itemize}

From claims (a) and  (b), we conclude the sets  $\Lambda_{+}$ and $\Lambda_{-}$ are both open and stable by positive and negative rays, respectively. It follows, by the connexity of $\mathbb R$, the set $D$ cannot be empty, and $D= [\lambda_{-}, \lambda_{+}]$ where $\lambda_{-} = \sup \Lambda_{-}$ and $\lambda_{+} = \inf \Lambda_{+}$.     
These facts prove item \ref{Item:ComplementPreferredOrdering_4}.

We only show the first claim, since the reasoning to obtain the other is analogous. 
Let $ \lambda \in \Lambda_{+} $. There exists $ \epsilon > 0 $ such that for any parameter $ \mu \in (\lambda-\epsilon, +\infty) $, the interaction $ E_\mu $ is non-degenerated.
As a matter of fact, if $ \big(z_0^n(\lambda), \ldots, z_n^n(\lambda)\big) $ is an increasing fundamental configuration for the interaction $ E_\lambda $, recalling that 
\[
E_0\big(z_0^n(\lambda), \ldots, z_n^n(\lambda)\big) - \lambda\big(z_n^n(\lambda) - z_0^n(\lambda)\big) = 
E_\lambda\big(z_0^n(\lambda), \ldots, z_n^n(\lambda)\big)  \le n \bar E_\lambda, 
\]
one has for $ \mu \in \mathbb R$,
\begin{align*}\label{Equation:ComplementPreferredOrdering_1}
 \bar E_\mu & \le \lim_{n\to +\infty} \frac{1}{n} E_\mu\big(z_0^n(\lambda), \ldots, z_n^n(\lambda)\big) \\
 & =  \lim_{n\to +\infty} \frac{1}{n} \big[ E_\lambda\big(z_0^n(\lambda), \ldots, z_n^n(\lambda)\big) 
 + (\lambda - \mu)\big(z_n^n(\lambda) - z_0^n(\lambda)\big) \big] \\
 & \le \bar E_\lambda + (\lambda - \mu)^{+} R(\lambda),
 \end{align*} 
where $ R(\lambda) > 0 $ is the uniform interdistance established by lemma~\ref{Lemma:AprioriEstimateFundamentalConfiguration}. 
It follows that the potential $E_{\mu}$ is non-degenerated for any  $\mu > \lambda - \epsilon$ where $ 0 < \epsilon \le  \left(\inf_{x} E_{0}(x,x)  - \bar E_{\lambda}\right)/R(\lambda)$.

It remains to show the preferred ordering for $ E_\mu $ is increasing.
For contradiction, let us assume the preferred ordering for $E_{\mu}$ is decreasing. 
Let $(z_{0}^{n}(\mu),\ldots, z_{n}^{n}(\mu))$  be a decreasing fundamental configuration for $E_{\mu}$.
Note that 
\begin{align*}
S_{\lambda}(z_{0}^{n}(\mu), z_{n}^{n}(\mu)) + & n \bar E_{\lambda} + (\lambda -\mu)(z_n^{n}(\mu) - z_{0}^{n}(\mu)) \le \\
& \le  E_{\lambda}(z_{0}^{n}(\mu),\ldots,z_{n}^{n}(\mu)) + (\lambda -\mu)(z_n^{n}(\mu) - z_{0}^{n}(\mu)) \\
& = E_{\mu}(z_{0}^{n}(\mu),\ldots,z_{n}^{n}(\mu)) \le n \bar E_{\mu}.
\end{align*}
Thanks to proposition~\ref{Proposition:FundamentalConfiguration}, there are $  \gamma(\lambda) > 0 $ and $\delta(\gamma) \ge 0 $ such that
\[
 (\gamma(\lambda) + \mu - \lambda) |z_n^{n}(\mu) - z_{0}^{n}(\mu)| - \delta(\lambda) \le n (\bar E_{\mu} - \bar E_{\lambda}) 
 \le n  (\lambda - \mu)^{+} R(\lambda).
\] 
This inequality does not hold for $ n $ large enough whenever $ \mu > \lambda $.
To obtain a contradiction also in the case $ \mu < \lambda $ with $ \mu $ near to $ \lambda $, 
we will need to look closely at the rotational properties of the fundamental configurations.
More precisely, up to consider a smaller $\epsilon$, we may assume that $\gamma(\lambda) / 2 \ge \epsilon$. 
For $ \mu \ge \lambda - \epsilon $, by passing to the limit as $ n $ tends to infinity in the above inequality,  we get
\[
\frac{\gamma(\lambda)}{2}  \phi(\mu) \le (\gamma(\lambda) - \epsilon) \phi(\mu) \le \epsilon R(\lambda).
\]
where $ \phi(\mu) > 0 $ is the uniform lower bound for a rotation number of any fundamental configuration for the
non-degenerated interaction $ E_\mu $ as established in lemma \ref{Lemma:AprioriEstimateFundamentalConfiguration}.
If such an inequality were valid, $ \lim_{\mu \to \lambda^{-}} \phi(\mu)~=~0 $ would occur.
However, we actually have $ \liminf_{\mu \to \lambda} \phi(\mu) > 0 $, from which we conclude our proof.
To see that, recall from~\eqref{definicao de eta}, lemma~\ref{lematecnico}, and lemma \ref{Lemma:AprioriEstimateFundamentalConfiguration}, that $\phi(\mu)$ is computed by the formula  
\begin{gather}
\phi(\mu) = \frac{\eta(\mu)^2}{\bar E_\mu - \inf_{x,y \in \mathbb{R}}E_\mu(x,y) +\eta(\mu)}, \label{Equation:ComplementPreferredOrdering_4}
\end{gather}
where $\eta(\mu)>0$ is any number fulfilling
\begin{gather}\label{Equation:ComplementPreferredOrdering_3}
\forall\, x, y \in\mathbb{R}, \, |y-x| < \eta(\mu) \ \Rightarrow \  E_\mu(x,y) -\bar E_\mu > \eta(\mu), 
\end{gather}
We first show that \eqref{Equation:ComplementPreferredOrdering_3} may be satisfied with an independent $ \eta_* := \min \big\{ R(\lambda), \frac{\eta(\lambda)}{2} \big\}$. Indeed, if $|x-y| < \eta_* $ and $|\mu-\lambda| \leq \frac{\eta(\lambda)}{4R(\lambda)}$,  then
\begin{align*}
E_\mu(x,y) - \bar E_\mu &\geq E_\lambda(x,y) - \bar E_\lambda - |\mu-\lambda| R(\lambda) -|\mu-\lambda||x-y| \\
&> \eta(\lambda) -2|\mu-\lambda| R(\lambda)  \geq \frac{\eta(\lambda)}{2} \geq \eta_*.
\end{align*}
We now bound from below $\phi(\mu)$ in \eqref{Equation:ComplementPreferredOrdering_4}.
The superlinearity of $E_{\lambda}$ implies that, given $ \epsilon_* := \frac{\eta(\lambda)}{4R(\lambda)} > 0 $, there is a constant $B_{\lambda} > 0$ such that
$$ \forall x,y \in \mathbb R, \, E_{\lambda} (x,y) -\bar E_{\lambda} \ge \epsilon_* |x-y| - B_{\lambda}.$$  
In particular, for every $|\mu-\lambda|\leq \epsilon_* $, one obtains for all $ x, y \in \mathbb{R} $,
\[
 E_\mu(x,y) \geq E_\lambda(x,y) - \epsilon_* |x-y| \geq \bar E_\lambda  - B_\lambda \geq \bar E_\mu  - \epsilon_* R(\lambda)  - B_\lambda, 
\]
so that $  \bar E_\mu - \inf_{x,y \in \mathbb{R}}E_\mu(x,y) \leq  \epsilon_* R(\lambda) + B_\lambda. $
Therefore, we have proved that
\begin{gather*}
\phi(\mu) \geq \frac{\eta_*^2}{\epsilon_* R (\lambda) + B_\lambda+ \eta_*} \qquad \text{if } |\mu - \lambda| \le \epsilon_*  = \frac{\eta(\lambda)}{4R(\lambda)}.
\end{gather*}
(If precision is required, the reader may easily check from the previous arguments that
\[
\epsilon := \min \left\{  \frac{\inf_{x} E_{0}(x,x)  - \bar E_{\lambda}}{R(\lambda)}, \frac{\gamma(\lambda)}{2}, \frac{\eta(\lambda)}{4R(\lambda)}, \frac{\gamma(\lambda)\eta_*^2}{4(\epsilon_* R (\lambda) + B_\lambda+ \eta_*)R(\lambda)} \right\} 
\]
could be proposed as an explicit definition to be considered in claim (a).)
\end{proof}

\appendix
\appendixpage
\addappheadtotoc

\section{Non-degenerate almost crystalline models}\label{Nao degenerescencia}

Our aim is to provide examples of pattern equivariant interactions that fulfill hypothesis
$ \inf_x E(x,x) > \bar E $.  For this purpose, we focus on one-dimensional quasicrystals studied in \cite{GambaudoGuiraudPetite2006}. More concretely, given $ \alpha \in (0,1/2) $ and $ \rho > 0 $, we will consider a  quasi-periodic set $ \omega(\alpha, \rho) = \omega = \{q_k\}_{k \in \mathbb Z} \subset {\mathbb R}$ defined by
\begin{gather*}
\forall\,  k \in \mathbb{Z}, \quad q_k = k + (\rho - 1) \lfloor k \alpha \rfloor.
\end{gather*}
Note that $q_k-q_{k-1} = 1+(\rho-1) a_k$, where 
\begin{gather*}
a_k = \lfloor k \alpha \rfloor - \lfloor (k-1) \alpha \rfloor.
\end{gather*}
Since $\alpha < 1 $, we have  $ (a_k)_{k \in\mathbb{Z}} \in \{0,1\}^{\mathbb Z}$ and $q_k-q_{k-1}$ equals $1$ or $\rho$ whenever $a_k$ equals $ 0 $ or $ 1 $, respectively.
As $ \alpha < 1/2 $, $ a_k$ and $a_{k+1}$ cannot be both equal to $1$. In fact,  $(a_k)_{k \in\mathbb{Z}}$ is  periodic if $\alpha$ is rational, and is called a Sturmian sequence  when  $\alpha$  is irrational~\cite{MorseHedlund1940}.

Since $\omega$  is uniformly discrete, it satisfies the \eqref{item:FLC} property. 
The fact that $ \omega $ obeys \eqref{item:Repet} can be assured essentially because irrational rotations on the circle  are minimal. 
Moreover, since  these rotations are also uniquely ergodic, they satisfy  an additional property: each type of pattern occurs with a positive density (see \cite{Fogg} for a modern presentation). 
In precise terms, they  fulfill
 
\begin{description}
 \item[uniform pattern distribution\label{item:UnifPatDistrib}] for any pattern $\texttt{P}$,  there is a positive number $\nu(\texttt{P})>0$ such that for any  nested sequence of  bounded open intervals $I_1 \subset I_2 \subset \ldots  \subset I_k \subset  \ldots$ with unbounded sequence of  lengths $(|I_k|)_k$, the quantity
\begin{equation*}
\frac{N_\texttt{P} (I_k)}{|I_k|} 
\end{equation*}
converges to $\nu(\texttt{P})$, where  $N_\texttt{P} (I)$ denotes the number of  patterns in $\omega\cap I$ equivalent to $\texttt{P}$. 
\end{description}

The uniform pattern distribution is a consequence of the uniform limits
\begin{equation}\label{frequencia simbolica}
\mu_{\mathbf b} = \lim_{n-m \to+\infty} \frac{\#_{\texttt b}a_{[m,n)}}{n-m},
\end{equation}
where $\#_{\texttt b}a_{[m,n)}$ is the number of times the subword $\texttt b = b_1 \cdots b_r \in \{0,1\}^r $, $r\geq1$, appears in the word $a_{[m,n)} = a_m a_{m+1} \cdots a_{n-1}$ for  $n>m$, $m,n \in\mathbb{Z}$.

A discrete set that satisfies  \eqref{item:FLC}, \eqref{item:Repet} and  \eqref{item:UnifPatDistrib} is said to be a \textit{quasicrystal}. See \cite{LagariasPleasant2003} for details.  

Let us explain how to define a $ C^2 $ interaction model $ E(x,y) $ that is twist and pattern equivariant with respect to $ \omega $. Regarding the interaction with the substrate, we introduce a pattern equivariant $ V $ obtained by translating two functions $ V_1 $ and $ V_\rho $
according to the patterns of two consecutive points of $ \omega $. Concretely, for $ V_1 : [0, 1] \to \mathbb R $
and $ V_\rho : [0, \rho] \to \mathbb R $ defined as $ V_1(x) = V_\rho(\rho x) / \rho^2 = 1 - \cos 2 \pi x $, consider for every $k \in\mathbb{Z}$ and $x \in [q_k,q_{k+1})$,
\[
V(x) = \left\{ \begin{array}{ll}
V_1(x-q_k), & \text{if $q_{k+1}-q_k=1$}, \\
V_\rho(x-q_k), & \text{if $q_{k+1}-q_k=\rho$}.
\end{array} \right.
\]
See the figure below for the graph of an example. Let the interaction be defined by
\begin{equation*}
E_\lambda (x,y) := \frac{1}{2} |y - x - \lambda|^2+KV(x),   \qquad \lambda \in \mathbb{R}, \ K \in \mathbb{R}_+. 
\end{equation*}


\begin{figure}[h]\label{figure}
\begin{center}
\begin{tikzpicture}[scale=0.7]
\draw[line width=1pt,->](-5,0)--(7,0); \draw(7.1,-0.2)node{$_x$};
\draw[ line width=0.5pt , ->](0,-0.5)--(0,2.5); 
\def\rhoscale{sqrt(3)}

\potentielL{-2-2*\rhoscale}{1} ;
\potentielL{-1-2*\rhoscale}{\rhoscale} ;
\potentielL{-1-\rhoscale}{1} ;
\potentielL{-\rhoscale}{\rhoscale} ;
\potentielL{0}{1} ;
\potentielL{1}{1} ;
\potentielL{2}{\rhoscale} ;
\potentielL{2+\rhoscale}{1} ;
\potentielL{3+\rhoscale}{\rhoscale} ;
\potentielL{3+2*\rhoscale}{1} ;

\foreach \t in {-3,...,5}{
	\draw({\t+(\rhoscale-1)*floor(\t/sqrt(5))},0)node{$\times$};
	\draw({\t+(\rhoscale-1)*floor(\t/sqrt(5))},-0.3)node{$q_{\t}$};
	}
\draw[<->] (0,-0.6)--(1,-0.6) node[midway,below] {$1$};
\draw[<->] ({2},-0.6)--({2+\rhoscale},-0.6) node[midway,below] {$\rho$}; 
  \end{tikzpicture}
  \end{center}
\caption{Graph of the potential $V$ for $\alpha=1/\sqrt{5}$ and $\rho = \sqrt{3}$}\label{fig:GraphePotentiel}
 \end{figure}


The next proposition  provides examples of  non-degenerate interactions for which theorem~\ref{theo:mainThm} does apply. 
Remark that for $\rho = 1$ the quasicrystal is nothing but the lattice $\mathbb Z$. 
For  $\rho \neq 1$, note that  the quasicrystal $\omega(\alpha, \rho)$ is linearly repetitive when the coefficients of the continued fraction of $\alpha$ are bounded \cite{Alisteetal} ({\it e.g.} when $\alpha$ is quadratic by the Lagrange's continued fraction theorem). It is known that the set of $\alpha$ satisfying this condition is a Baire meager set, of zero Lebesgue measure but with Hausdorff dimension $1$.  

\begin{proposition}\label{Exemplo hipotese}
For $ 0 < \rho < 1 + \alpha^{-1/2} $, denote $c_1(\lambda) := \frac{ \lambda^2}{8} \frac{1-\alpha(1-\rho)^2}{1+\alpha(\rho^2-1)}$. Then for every $K \in (0, c_1(\lambda))$, 
$$ \inf_{x \in \mathbb R} E_\lambda(x, x) > \bar E_\lambda. $$
\end{proposition}

In the periodic case, that is, when $\rho=1$, the proposition gives a simple estimate between $\lambda$ and $K$ for the non-degeneracy. For $K < \frac{\lambda^2}{8}$, the periodic Frenkel-Kontorova model is non-degenerate.

\begin{proof}
Obviously $ \inf_x E_\lambda(x,x) = \lambda^2/2 $. 
Since
\begin{gather*}
\bar E_\lambda = \lim_{n\to\infty} \inf_{x_0, \ldots, x_n}  \frac{1}{n} E_\lambda(x_0, \ldots, x_n), 
\end{gather*}
the strategy consists in finding a long configuration by concatenating short configurations located  in each cell $[q_k,q_{k+1}]$ so that the mean of the successive interactions is strictly smaller than $ \lambda^2/2 $.
Let $ \ell \ge 1 $. 
For cells $(q_k,q_{k+1})$ of length $q_{k+1}-q_k = 1$, we use a translate $q_k+\mathbf y^\ell$ of the subconfiguration $ \mathbf y^\ell$ in $(0,1)$
\begin{gather*}
\mathbf{y}^\ell = \big(y_0^\ell, \ldots, y_{2^\ell-1}^\ell\big) := \Big( \frac{1}{2^{\ell+1}} + \frac{j}{2^\ell} \Big)_{j=0}^{2^\ell-1}. 
\end{gather*}
For cells of length $q_{k+1}-q_k=\rho$, we use a translate of the homothetic subconfiguration $ \rho \mathbf{y}^\ell = (\rho y_0^\ell, \ldots, \rho y_{2^\ell-1}^\ell) $.
More precisely, we define for every $k\geq1$ a configuration $\mathbf x_k^\ell$ of $k 2^\ell + 1$ points obtained by concatenating $k$ translates of some homothetic  $\mathbf y^\ell$ and by adding an extra translate of $y_0^\ell$ 
\begin{multline*}
\mathbf x_k^\ell := \big( (q_1-q_0)\mathbf y^\ell,  q_1 + (q_2-q_1)\mathbf y^\ell, \ldots,  q_{k-1}+(q_k - q_{k-1})\mathbf y^\ell,  \\
q_k + (q_{k+1}-q_k)y_0^\ell \big) \in \mathbb R^{k 2^\ell + 1}.
\end{multline*}
Note that $\mathbf x_k^\ell$ has $2^\ell$ entries belonging to each interval $(q_0,q_1), \ldots,  (q_{k-1},q_k)$, and an extra point of $(q_k,q_{k+1})$. 
The total energy $E_\lambda(\mathbf x_k^\ell)$ consists in two terms. 
The term coming from the external potential is a sum of partial sums,  of the form either $\sum_{x \in\mathbf{y}^\ell + q} V_1(x) $ or  $ \sum_{x \in\rho \mathbf{y}^\ell + q} V_\rho(x) $. The main observation is that in both cases each partial sum boils down to 
$$ c\big(2^\ell - \sum_{j=0}^{2^\ell-1} \cos(2 \pi y_j^\ell + d)\big), \qquad \text{ with $ c, d $ constants}, \, c \in \{ 1, \rho^2\}. $$
Since $\sum_{j=0}^{2^\ell-1} \cos (2\pi y_j^\ell) = \sum_{j=0}^{2^\ell-1} \sin (2\pi y_j^\ell) = 0$, we thus have
$$ \sum_{x \text{ entry into }  \mathbf{y}^\ell + q} V_1(x) =2^\ell \quad \text{and}  \sum_{x \text{ entry into }  \rho\mathbf{y}^\ell + q} V_\rho(x) =  \rho^2 2^\ell.  $$
The term coming from the mutual interaction between neighboring atoms can be calculated according to the cases of entries belonging to possible cells, and presents three values. In fact, let  $ x < y $ be consecutive entries  into $ \mathbf{x}_k^\ell $. The energy $ \frac{1}{2} |y - x - \lambda|^2 $ 
takes one of the following values:
\begin{align*}
\frac{\lambda^2}2+ \frac{\rho}{2^\ell} \Big( \frac{\rho}{2^{\ell+1}} - \lambda \Big) & \quad \textrm{ if } x,y \in [q_{k-1},q_k], \text{ with } a_k =  1, \\
\frac{\lambda^2}2+ \frac{1}{2^\ell} \Big( \frac{1}{2^{\ell+1}} - \lambda \Big) 
& \quad  \textrm{ if } x,y \in [q_{k-1},q_{k+1}], \text{  with }  a_k = a_{k+1} =0, \\ 
\frac{\lambda^2}2+ \frac{\rho +1}{2^{\ell+1}} \Big( \frac{\rho+1}{2^{\ell+2}} - \lambda \Big)   & \quad \textrm{ if } x \in [q_{k-1},q_k] \text{ and }  y  \in [q_k,q_{k+1}], \\
& \quad \text{ with }  a_k \neq a_{k+1} 
\end{align*}

For $k\geq1$, let $\#_{i}^k$ denote the number of times the subword  $i \in \{0, 1\}$ appears in the word $a_1 \cdots a_{k+1}$. 
Only three types of subwords of length 2 appear in $ (a_k)_{k \in \mathbb Z} $: $00$, $01$ and $10$.  Introduce similarly $\#_{ij}^k$ as the number of times the subword $ij$, $i, j \in \{0,1\}$, appears in the word $a_1 \cdots a_{k+1}$. 
Clearly $\#_{11}^k=0$.  Then
\begin{align*}
 E_\lambda&(\mathbf x_k^\ell)  :=  K 2^\ell (\#_0^k+ \rho^2\#_1^k) + \#_{00}^k 2^\ell \Big[ \frac{\lambda^2}{2} + \frac{1}{2^\ell} \Big( \frac{1}{2^{\ell+1}} -\lambda \Big) \Big] \\
&+ \#_{01}^k (2^\ell-1) \Big[  \frac{\lambda^2}{2} + \frac{1}{2^\ell} \Big( \frac{1}{2^{\ell+1}}-\lambda \Big) \Big] +  \#_{01}^k  \Big[  \frac{\lambda^2}{2} + \frac{\rho+1}{2^{\ell+1}} \Big( \frac{\rho+1}{2^{\ell+2}}-\lambda \Big) \Big] \\
&+ \#_{10}^k (2^\ell-1) \Big[  \frac{\lambda^2}{2} + \frac{\rho}{2^\ell} \Big( \frac{\rho}{2^{\ell+1}} -\lambda \Big) \Big] +  \#_{10}^k  \Big[  \frac{\lambda^2}{2} + \frac{\rho+1}{2^{\ell+1}} \Big( \frac{\rho+1}{2^{\ell+2}} -\lambda \Big) \Big].
\end{align*} 
We divide by $k2^\ell$ and let $k\to+\infty$. From~\eqref{frequencia simbolica}, we have $\#_{0}^k/k \to \mu_{0}$ and a similar result for $ \#_{1}^k$, $\#_{00}^k$, $\#_{01}^k$ and $\#_{10}^k$. 
Since the word $1$ is always preceded and followed by $0$  in $(a_k)_k$, note then that $\mu_1=\mu_{01} = \mu_{10} = (1-\mu_{00})/2$. 
We thus obtain
\begin{align*}
 \lim_{k\to+\infty} \frac{E_\lambda(\mathbf x_k^\ell)}{k2^\ell} = \frac{\lambda^2}{2}+  K(1+&\mu_1(\rho^2-1)) + \frac{1+\mu_{00}}{2}  \frac{1}{2^\ell} \Big( \frac{1}{2^{\ell+1}} -\lambda \Big)  \\ 
& \, +  \frac{1-\mu_{00}}{2}  \frac{\rho}{2^{\ell}} \Big( \frac{\rho}{2^{\ell+1}} -\lambda \Big)  - \frac{1-\mu_{00}}{2^\ell}  \frac{(1-\rho)^2}{2^{2\ell+3}}.
\end{align*}
We want to find a condition on $\lambda>0$ so that $\bar E_\lambda < \frac{\lambda^2}{2}$. Ignoring the nonpositive term, it is enough to choose $\lambda$ so that
\[
2^\ell K(1+\mu_1(\rho^2-1))+ \frac{1+\mu_{00}}{2}  \Big( \frac{1}{2^{\ell+1}} -\lambda \Big) +  \frac{1-\mu_{00}}{2}  \rho \Big( \frac{\rho}{2^{\ell+1}} -\lambda \Big) < 0, 
\]
or equivalently 
\begin{equation*} 
 2^\ell K(1+\mu_1(\rho^2-1)) + \frac{1+\mu_{00}}{2}  \frac{1}{2^{\ell+1}} +  \frac{1-\mu_{00}}{2} \frac{\rho^2}{2^{\ell+1}} < \lambda \Big( \frac{1+\mu_{00}}{2} + \rho \frac{1-\mu_{00}}{2} \Big).
\end{equation*}
Assume $\ell$ has been chosen so that $\frac{1}{2^\ell} < \lambda \leq  \frac{1}{2^{\ell-1}}$. It suffices to rewrite the inequality above with $\frac{1}{2^\ell}$ instead of $\lambda$. Then $K$ must satisfy
\begin{gather*}
2^\ell K(1+\mu_1(\rho^2-1))  <  \frac{1+\mu_{00}}{2} \frac{1}{2^{\ell+1}} + \frac{1-\mu_{00}}{2} \frac{\rho(2-\rho)}{2^{\ell+1}}.
\end{gather*}
As $\lambda \leq \frac{1}{2^{\ell-1}}$, we have $\frac{\lambda^2}{8} \leq \frac{1}{2^{2\ell+1}}$. 
Hence, it is enough to choose $K$ so that
\begin{gather*}
K < \frac{\lambda^2}{8}\Big( \frac{1+\mu_{00}}{2}  +  \frac{1-\mu_{00}}{2} \rho(2-\rho) \Big) \frac{1}{1+\mu_1(\rho^2-1)}.
\end{gather*}
Recalling that $(1-\mu_{00})/2 =\mu_1$, we have
\begin{gather*}
\frac{1+\mu_{00}}{2}  +  \frac{1-\mu_{00}}{2} \rho(2-\rho) = 1-\mu_1 + \mu_1 \rho(2-\rho) = 1-\mu_1(1-\rho)^2.
\end{gather*}
We thus have proved that $\bar E_\lambda < \frac{\lambda^2}{2}$ whenever $K < \frac{\lambda^2}{8}\frac{1-\mu_1(1-\rho)^2}{1+\mu_1(\rho^2-1)}$.
A direct computation \cite{MorseHedlund1940} shows that $\mu_1 = \alpha$ . This proves the proposition.  
\end{proof}

\section{The calibration level of the sublinearity}\label{nivel unico}

Here we show that, within the framework delimited by hypotheses~\ref{assumption:interactionModel} and definition~\ref{invariancia de motivos}, the only calibration level associated with sublinear weak KAM solutions is exactly that given by the ground action $ \bar E $. This result is nothing unexpected. We present the argument, however, for the convenience of the reader. We obviously make use of the Lax-Oleinik operator $ T $ as introduced in~\eqref{definition:LaxOleinikOperator}. Furthermore, by sublinearity we mean the notion described by~\eqref{Equation:StrictlySublinearVariation}. The precise statement is the following one.

\begin{proposition}\label{nivel chave}
Let $ E $ be a weakly twist interaction that is pattern equivariant 
with respect to a quasi-periodic set. Suppose that $ \inf_x E(x,x) > \bar E $.
Let $ u $ be a sublinear function such that $ T[u] = u + c $ for some constant $ c $.
Then $ c = \bar E $.
\end{proposition}

The fact that $u$ is a fixed point for the Lax-Oleinik operator associated with the interaction $ E - c $ determines its regularity.

\begin{lemma}
Let $ E $ be an interaction model satisfying assumptions (H1-3) of hypotheses~\ref{assumption:interactionModel}. If  $u$ is a function satisfying  $ T[u] = u + c $ for some $ c$, then there are $K_c,R_c\geq0$ (depending only on the interaction model $E$ and on the constant $c$) such that
\begin{enumerate}
\item  $u$ is Lipschitz continuous with $\text{\rm Lip}(u) \leq K_c$,
\item  $\forall\, y \in \mathbb{R}, \, \, \argmin \{ u(\cdot) + E(\cdot,y) \} \subset [y-R_c, y+R_c]$.
\end{enumerate}
\end{lemma}

\begin{proof}
Just repeat {\it verbatim} the proof of lemma~\ref{Lemma:AprioriBoundWeakKAMsolution} by taking into account $ \tilde K = \sup_{|y-x| \leq 1} E(x,y) - c $.
\end{proof}

\begin{proof}[Proof of proposition~\ref{nivel chave}]
We first show the inequality $ c \ge \bar E $. To that end, let $ (y_0, \ldots, y_n) $ be a configuration such that
$$ u(y_n) - u(y_0) = E(y_0, \ldots, y_n) - n c. $$
By the sublinearity of $ u $, given $ \alpha > 0 $, there is $ \beta $ with respect to which 
$ | u(y_n) - u(y_0) | \le \alpha |y_n - y_0| + \beta $. Therefore, thanks to the previous lemma, we obtain 
$$ c \ge - \frac{\alpha n R_c + \beta}{n} + \frac{1}{n} \inf_{(x_0, \ldots, x_n)} E(x_0, \ldots, x_n), $$
and the claim follows by passing to limit as $ n $ tends to infinity and then by considering $ \alpha $ arbitrarily close to zero.

We now show that $ \bar E \ge c $. Let $ (z_0, \ldots, z_n) $ be a fundamental configuration of size $ n $ (recall definition~\ref{Definition:FundamentalConfiguration}). We thus have
\begin{gather*}
E(z_0,\ldots, z_n)  \leq n\bar E, \quad \text{and} \\
u(z_n) - u(z_0) \le E(z_0, \ldots, z_n) - n c.
\end{gather*}
We then apply lemma~\ref{Lemma:AprioriEstimateFundamentalConfiguration} and again the sublinearity of $ u $ to get
$$ c \le \bar E + \frac{\alpha n R + \beta}{n}, $$
from which the conclusion follows as for the converse inequality. 
\end{proof}

\end{document}